\documentclass{lmcs} %%% last changed 2014-08-20
\pdfoutput=1
\usepackage[utf8]{inputenc}

\usepackage{lastpage}
\lmcsdoi{20}{1}{6}
\lmcsheading{}{\pageref{LastPage}}{}{}%
{Mar.~08,~2023}{Jan.~23,~2024}{}


\usepackage{amssymb}
\usepackage{stmaryrd}
\usepackage{mathtools,centernot}
\usepackage{xcolor,xargs,xspace}
\usepackage{hyperref}
\usepackage{booktabs}
\usepackage{tikz}
\usetikzlibrary{arrows,positioning,decorations.markings}
\tikzset{D/.style={>=stealth'}}
\tikzset{Label/.style={scale=0.8}}
\tikzset{-||->/.style={postaction={decorate},decoration={%
markings,%
mark=at position 0.5 with {\node[sloped,rotate=90,transform shape] {=};}}}}

\newcommand{\pref}[1]{\ref{#1}}
\newcommand{\cref}[1]{\autoref{#1}}
\newcommand{\secref}[1]{Section~\ref{#1}}

\newcommand{\m}[1]{\mathsf{#1}}

\newcommand{\seq}[2][n]{{#2_1},\dots,{#2_{#1}}}

\renewcommand{\AA}{\mathcal{A}}
\newcommand{\BB}{\mathcal{B}}
\newcommand{\CC}{\mathcal{C}}
\newcommand{\DD}{\mathcal{D}}
\newcommand{\FF}{\mathcal{F}}
\newcommand{\NN}{\mathbb{N}}
\newcommand{\PP}{\mathcal{P}}
\newcommand{\RR}{\mathcal{R}}
\newcommand{\CCd}{\mathcal{C}_\mathsf{d}}
\renewcommand{\SS}{\mathcal{S}}
\newcommand{\SSS}{\mathsf{S}}
\newcommand{\TT}{\mathcal{T}}
\newcommand{\VV}{\mathcal{V}}
\newcommand{\Dom}{\mathcal{D}\mathsf{om}}
\newcommand{\CPS}{\mathsf{CPS}}
\newcommand{\PCPS}{\mathsf{PCPS}}
\newcommand{\NaTT}{\textsf{NaTT}\xspace}
\newcommand{\ACP}{\textsf{ACP}\xspace}

\newcommand{\CoLLSaigawa}{\textsf{CoLL-Saigawa}\xspace}
\newcommand{\CoLL}{\textsf{CoLL}\xspace}

\newcommand{\CSI}{\textsf{CSI}\xspace}

\newcommand{\OURTOOL}{\textsf{Hakusan}\xspace}

\newcommand{\cp}[1]{%
\mathrel{{_{#1}{\xleftarrow{}}{\rtimes}{\xrightarrow{\epsilon}_{#1}}}}}

\newcommand{\cpT}[2]{%
\mathrel{{_{#1}{\xleftarrow{#2}}{\rtimes}{\xrightarrow{\epsilon}_{#1}}}}}
\newcommand{\Fun}{\mathcal{F}\mathsf{un}}
\newcommand{\Pos}{\mathcal{P}\mathsf{os}}
\newcommand{\Var}{\mathcal{V}\mathsf{ar}}

\newcommand{\pcp}[1]{\mathrel{{\pfromB{}{#1}}{\rtimes}{\xrightarrow{\epsilon}_{#1}}}}

\newcommand{\from}{\mathrel{\leftarrow}}
\newcommand{\fromto}{\mathrel{\xleftrightarrow{}}}
\newcommand{\fromB}[1]{\mathrel{_{#1}{\from}}}
\newcommand{\fromT}[1]{\mathrel{^#1{\from}}}
\newcommand{\fromBT}[2]{\mathrel{\prescript{#2}{#1}{\from}}}
\newcommand{\fromBss}[2]{\mathrel{_{#1}\xleftarrow{\smash{#2}}}}

\newcommand{\To}{\mathrel{\Rightarrow}}
\newcommand{\From}{\mathrel{\Leftarrow}}
\newcommand{\FromB}[1]{\mathrel{_{#1}{\From}}}

\newcommand{\Fromto}{\xLeftrightarrow{}}

\newcommand{\rto}{\xrightarrow{\smash{\epsilon}}}

\makeatletter
\newcommand*{\narrowfill@}[5]{%
 $\m@th\thickmuskip0mu\medmuskip\thickmuskip\thinmuskip\thickmuskip
 \relax#5#1\mkern-7mu%
 \cleaders\hbox{$#5\mkern-2mu#2\mkern-2mu$}\hfill
 \mkern-5mu%
 #4%
 \mkern-5mu%
 \cleaders\hbox{$#5\mkern-2mu#2\mkern-2mu$}\hfill
 \mkern-7mu#3$%
}
\newcommand*{\ptofill@}{%
 \narrowfill@\relbar\relbar\rightarrow%
 {\relbar\mkern-10mu\mapstochar\mkern7mu\mapsfromchar\mkern3mu}
}
\newcommand*{\pto}[2][]{%
 \ext@arrow 0359\ptofill@{#1}{#2}%
}
\newcommand*{\pfromfill@}{%
  \narrowfill@\leftarrow\relbar\relbar
  {\relbar\mkern-11mu\mapstochar\mkern7mu\mapsfromchar\mkern3mu}
}
\newcommand*{\pfrom}[2][]{%
 \ext@arrow 3359\pfromfill@{#1}{#2}%
}
\newcommand*{\pfromtofill@}{%
  \narrowfill@\leftarrow\relbar\rightarrow
  {\relbar\mkern-11mu\mapstochar\mkern7mu\mapsfromchar\mkern3mu}
}
\newcommand*{\pfromto}[2][]{%
 \ext@arrow 3359\pfromtofill@{#1}{#2}%
}
\newcommand*{\mtofill@}{%
 \narrowfill@\relbar\relbar\rightarrow%
 {\relbar\mkern-10mu\circ\mkern-14mu\relbar\mkern3mu}
}
\newcommand*{\mto}[2][]{%
 \ext@arrow 0359\mtofill@{#1}{#2}%
}
\newcommand*{\mfromfill@}{%
  \narrowfill@\leftarrow\relbar\relbar
  {\relbar\mkern-11mu\circ\mkern-11mu\relbar\mkern3mu}
}
\newcommand*{\mfrom}[2][]{%
 \ext@arrow 3359\mfromfill@{#1}{#2}%
}
\makeatother
\newcommand{\pfromB}[2]{\mathrel{_{#2}{\pfrom{\smash{#1}}}}}

\newcommand{\criterion}[1]{\textbf{\sffamily#1}}

\theoremstyle{thmC}
\newtheorem{corC}[thm]{Corollary}

\begin{document}

\title{Compositional Confluence Criteria}
% \titlecomment{{\lsuper*}OPTIONAL comment concerning the title, \eg,
%   if a variant or an extended abstract of the paper has appeared elsewhere.}
\thanks{The research described in this paper is supported
by JSPS KAKENHI Grant Numbers JP22K11900.}%optional

% affiliations are numbered automatically with a, b, c (see below)
% use the optional argument to indicate the affiliation(s) of each author
% omit the argument if there is only one author, or only one affiliation
\author[K.~Shintani]{Kiraku Shintani\lmcsorcid{https://orcid.org/0000-0002-2986-4326}}
\author[N.~Hirokawa]{Nao Hirokawa\lmcsorcid{https://orcid.org/0000-0002-8499-0501}}

% affiliation 1 (automatically numbered a)
\address{JAIST, Japan}	%optional
% write emails for all authors having that affiliation
\email{s1820017@jaist.ac.jp, hirokawa@jaist.ac.jp}  %optional

%% etc.

%% required for running head on odd and even pages, use suitable
%% abbreviations in case of long titles and many authors:

%%%%%%%%%%%%%%%%%%%%%%%%%%%%%%%%%%%%%%%%%%%%%%%%%%%%%%%%%%%%%%%%%%%%%%%%%%%

%% the abstract has to PRECEDE the command \maketitle:
%% be sure not to issue the \maketitle command twice!

\begin{abstract}
\noindent We show how confluence criteria based on decreasing diagrams are
generalized to ones composable with other criteria.  For demonstration of the
method, the confluence criteria of orthogonality, rule labeling, and critical
pair systems for term rewriting are recast into composable forms.
We also show how such a criterion can be used for a reduction
method that removes rewrite rules unnecessary for confluence analysis. In
addition to them, we prove that Toyama's parallel closedness result based
on parallel critical pairs subsumes his almost parallel closedness theorem.
\end{abstract}

\maketitle

\section{Introduction}

Confluence is a property of rewriting that ensures uniqueness of
computation results.  In the last decades, 
various proof methods for confluence of term rewrite systems
have been developed. They are roughly classified to three groups:
(direct) confluence criteria based on critical pair
analysis~\cite{KB70,H80,T81,T88,G96,vO97,O98,vO08,ZFM15},
decomposition methods based on modularity and
commutation~\cite{T87,AYT09,SH15},
and transformation methods based on 
simulation of rewriting~\cite{AT12,K95,NFM15,SH15}.

In this paper we present a confluence analysis based on \emph{compositional}
confluence criteria.  Here a compositional criterion means a sufficient
condition that, given a rewrite system $\RR$ and its subsystem
$\CC \subseteq \RR$, confluence of $\CC$ implies that of $\RR$.  Since
such a subsystem can be analyzed by any other (compositional)
confluence criterion, compositional criteria can be seen as a combination
method for confluence analysis.  Because the empty system is confluent, 
by taking the empty subsystem $\CC$ compositional criteria can be
used as ordinary (direct) confluence criteria. 

In order to develop compositional confluence criteria we revisit van Oostrom's
decreasing diagram technique~\cite{vO94,vO08}, which is known as a powerful
confluence criterion for abstract rewrite systems.  Most existing confluence
criteria for left-linear rewrite systems, including the ones listed above, can be
proved by decreasingness of parallel steps or multi-steps.  Recasting the
decreasing diagram technique as a compositional criterion, we demonstrate
how confluence criteria based on decreasing diagrams can be reformulated as
compositional versions.  We pick up the confluence criteria by
orthogonality~\cite{R73}, rule labeling~\cite{ZFM15}, and critical pair
systems~\cite{HM11}.

As mentioned above, compositional confluence criteria guarantee that
confluence of a subsystem implies confluence of the original rewrite
system.  If the converse also holds, confluence of $\RR$ is equivalent to
that of $\CC$.  In other words, we may reduce the confluence problem of
$\RR$ to that of the subsystem $\CC$, without assuming confluence of the
latter. Such a \emph{reduction method} is useful when analyzing confluence
automatically. We present a simple method inspired by redundant rule
elimination techniques~\cite{SH15,NFM15}.

In addition to them, we elucidate the hierarchy of Toyama's two parallel
closedness theorems~\cite{T81,T88} and rule labeling based on parallel
critical pairs~\cite{ZFM15}.  As a consequence, it turns out that rule
labeling and its compositional version are generalizations of Huet's and
Toyama's (almost) parallel closedness theorems.

The remaining part of the paper is organized as follows:
In \secref{sec:preliminaries} we recall notions from rewriting.
In \secref{sec:pc} we show that Toyama's almost parallel closedness is
subsumed by his earlier result based on parallel critical pairs.  In
\secref{sec:ddc}, we introduce an abstract criterion for 
our approach, and in the subsequent three sections we 
derive compositional criteria from the confluence criteria of
orthogonality (\secref{sec:orthogonality}),
rule labeling (\secref{sec:prlc}), and
the criterion by critical pair systems (\secref{sec:cpsc}).
In \secref{sec:reduction} we present a non-confluence criterion that
strengthens compositional confluence criteria to a reduction method.
\secref{sec:experiments} reports experimental results.
Discussing related work and potential future work in
\secref{sec:related-work}, we conclude the paper.

A preliminary version of this paper appeared in the proceedings of the
7th International Conference on Formal Structures for Computation and
Deduction~\cite{SH22}.  Compared with it, the reduction method
presented in \secref{sec:reduction} is a new result and the
experimental evaluation has been extended.
Moreover, the present paper includes a complete proof for a key lemma
(\cref{lem:PML}(b)) for confluence analysis based on parallel critical
pairs.  The lemma itself is known~\cite{G96,ZFM15} but its proof 
is not presented in the literature.

\section{Preliminaries}
\label{sec:preliminaries}

Throughout the paper, we assume familiarity with abstract rewriting and
term rewriting~\cite{BN98,TeReSe}.  
We just recall some basic notions and notations for rewriting
and confluence.

An ($I$-indexed) \emph{abstract rewrite system} (ARS) $\AA$ is a pair 
$(A, \{\to_\alpha\}_{\alpha \in I})$ consisting of a set $A$ and a family
of relations $\to_\alpha$ on $A$ for all $\alpha \in I$.  Given a subset
$J$ of $I$, we write $x \to_J y$ if $x \to_\alpha y$ for some index 
$\alpha \in J$.  The relation $\to_I$ is referred to as $\to_\AA$.
An ARS $\AA$ is called \emph{confluent} or \emph{locally confluent} if
\(
{\fromBT{\AA}{*} \cdot \to_\AA^*} \subseteq
{\to_\AA^* \cdot \fromBT{\AA}{*}}
\)
or 
\(
{\fromB{\AA}{} \cdot \to_\AA} \subseteq
{\to_\AA^* \cdot \fromBT{\AA}{*}}
\)
holds, respectively.  We say that ARSs $\AA$ and $\BB$ \emph{commute} if
\(
{\fromBT{\AA}{*} \cdot \to_\BB^*} \subseteq
{\to_\BB^* \cdot \fromBT{\AA}{*}}
\)
holds.  A conversion of form $b \fromB{\AA}{} a \to_{\BB} c$ is called
a \emph{local peak} (or simply a \emph{peak}) between $\AA$ and $\BB$.
A relation $\to$ is \emph{terminating} if 
there exists no infinite sequence $a_0 \to a_1 \to \cdots$.
We say that an ARS $\AA$ is \emph{terminating} if $\to_\AA$ is terminating.
We define $\to_{\AA/\BB}$ as $\to_\BB^* \cdot \to_\AA^{} \cdot \to_\BB^*$.  We
say that $\AA$ is \emph{relatively terminating} with respect to $\BB$, or
simply $\AA/\BB$ is \emph{terminating}, if $\to_{\AA/\BB}$ is terminating.

Positions are sequences of positive integers.  The empty sequence
$\epsilon$ is called the \emph{root} position.  We write $p \cdot q$ or
simply $pq$ for the concatenation of positions $p$ and $q$.  The prefix
order $\leqslant$ on positions is defined as $p \leqslant q$ if 
$p \cdot p' = q$ for some $p'$.  We say that positions $p$ and $q$ are
\emph{parallel} if $p \nleqslant q$ and $q \nleqslant p$.
A set of positions is called \emph{parallel} if all its elements
are so.

Terms are built from a signature $\FF$ and a countable set $\VV$ of
variables satisfying $\FF \cap \VV = \varnothing$.  
The set of all terms (over $\FF$) is denoted by $\TT(\FF,\VV)$. 
Let $t$ be a term.
The set of all variables in $t$ is denoted by $\Var(t)$, and the set of
all function symbols in a term $t$ by $\Fun(t)$.
The set of all function positions and the set of variable positions in $t$
are denoted by $\Pos_\FF(t)$ and $\Pos_\VV(t)$, respectively.
The \emph{subterm} of $t$ at position $p$ is denoted by $t|_p$.  It is a
\emph{proper} subterm if $p \neq \epsilon$.  By $t[u]_p$ we denote the term
that results from replacing the subterm of $t$ at $p$ by a term $u$.  The
size $|t|$ of $t$ is the number of occurrences of functions symbols and
variables in $t$.  A term $t$ is said to be \emph{linear} if every variable
in $t$ occurs exactly once.  

A \emph{substitution} is a mapping $\sigma : \VV \to \TT(\FF,\VV)$
whose \emph{domain} $\Dom(\sigma)$ is finite.  Here $\Dom(\sigma)$ stands
for the set $\{ x \in \VV \mid \sigma(x) \neq x\}$.  The term $t\sigma$ is
defined as $\sigma(t)$ for $t \in \VV$, and $f(t_1\sigma,\dots,t_n\sigma)$
for $t = f(\seq{t})$. 
A term $u$ is called an \emph{instance} of $t$ if $u = t\sigma$
for some $\sigma$.
A substitution is called a \emph{renaming} if it is a bijection on variables.
The \emph{composition} $\sigma\tau$ of two substitutions $\sigma$ and
$\tau$ is defined by $(\sigma\tau)(x) = (x\sigma)\tau$.  An equation is a
pair $(s, t)$ of terms, written as $s \approx t$.  Let $E$ be a set of
equations.  A substitution $\sigma$ is said to be a \emph{unifier} of a set
$E$ of equations if $s\sigma = t\sigma$ holds for all $s \approx t \in E$.
A unifier $\sigma$ of $E$ is \emph{most general} if for every unifier
$\tau$ of $E$ there exists a substitution $\sigma'$ such that 
$\tau = \sigma\sigma'$.
A unifier of $\{ s \approx t \}$ is
said to be
a unifier of $s$ and $t$.

A \emph{term rewrite system} (TRS) over $\FF$ is a set of rewrite rules.
Here a pair $(\ell, r)$ of terms over $\FF$ is a \emph{rewrite rule}
or simply a \emph{rule} if $\ell \notin \VV$ and $\Var(r) \subseteq \Var(\ell)$.
We denote it by $\ell \to r$. The rewrite relation $\to_\RR$ of a TRS $\RR$
is defined on terms as follows: $s \to_\RR t$ if $s|_p = \ell\sigma$ and $t
= s[r\sigma]_p$ for some rule $\ell \to r \in \RR$, position $p$, and
substitution $\sigma$.  We write $s \xrightarrow{\smash{p}}_\RR t$
if the rewrite position $p$ is relevant.
We call subsets of $\RR$ \emph{subsystems}.
We write $\Fun(\ell \to r)$ for $\Fun(\ell) \cup \Fun(r)$ and $\Fun(\RR)$
for the union of $\Fun(\ell \to r)$ for all rules $\ell \to r \in \RR$.
The set $\{ f \mid f(\seq{\ell}) \to r \in \RR \}$ 
is the set of \emph{defined symbols} and denoted by $\DD_\RR$.
A TRS $\RR$ is \emph{left-linear} if $\ell$ is linear for all
$\ell \to r \in \RR$.  Since any TRS $\RR$ can be regarded as the ARS
$(\TT(\FF,\VV), \{ \to_\RR \})$, we use notions and notations of ARSs for
TRSs.  For instance, a TRS $\RR$ is (locally) confluent if the ARS
$(\TT(\FF,\VV), \{ \to_\RR \})$ is so.  Similarly, two TRSs commute if
their corresponding ARSs commute.

Local confluence of TRSs is characterized by the notion of critical pair.
We say that a rule $\ell_1 \to r_1$ is a \emph{variant} of
a rule $\ell_2 \to r_2$ if $\ell_1\rho = \ell_2$ and $r_1\rho = r_2$ for
some renaming $\rho$.

\begin{defi}
Let $\RR$ and $\SS$ be TRSs.  Suppose that the following conditions hold:
\begin{itemize}
\item
$\ell_1 \to r_1$ and $\ell_2 \to r_2$ are variants of rules in $\RR$ and
in $\SS$, respectively,
\item
$\ell_1 \to r_1$ and $\ell_2 \to r_2$ have no common variables,
\item
$p \in \Pos_\FF(\ell_2)$,
\item
$\sigma$ is a most general unifier of $\ell_1$ and $\ell_2|_p$, and
\item
if 
$p = \epsilon$ then $\ell_1 \to r_1$ is not a variant of 
$\ell_2 \to r_2$.
\end{itemize}
The local peak
\(
(\ell_2\sigma)[r_1\sigma]_p \fromBss{\RR}{p} \ell_2\sigma \rto_\SS r_2\sigma
\)
is called a \emph{critical peak} between $\RR$ and $\SS$.  When 
$t \fromBss{\RR}{p} s \rto_\SS u$ is a critical peak, the pair $(t, u)$
is 
called a \emph{critical pair}. To clarify the orientation of
the pair, we denote it as the binary relation
$t \mathrel{{\fromBss{\RR}{p}}{\rtimes}{\rto_\SS}} u$, see~\cite{D05}.
Moreover, we write
$t \mathrel{{\fromB{\RR}}{\rtimes}{\rto_\SS}} u$ if
$t \mathrel{{\fromBss{\RR}{p}}{\rtimes}{\rto_\SS}} u$
for some position $p$.
\end{defi}
\begin{thmC}[\cite{H80}]
A TRS $\RR$ is locally confluent if and only if
${\cp{\RR}} \subseteq {\to^*_\RR \cdot \fromBT{\RR}{*}}$
holds.
\end{thmC}

Combining it with Newman's Lemma~\cite{N42}, we obtain 
Knuth and Bendix' criterion~\cite{KB70}.

\begin{thmC}[\cite{KB70}]
\label{thm:KB70}
A terminating TRS $\RR$ is confluent if and only if
the inclusion
${\cp{\RR}} \subseteq {\to^*_\RR \cdot \fromBT{\RR}{*}}$
holds.
\end{thmC}

We define the parallel step relation, which plays a key role in
analysis of local peaks.

\begin{defi}
Let $\RR$ be a TRS and let $P$ be a set of parallel positions.  The
\emph{parallel step} $\pto{P}_\RR$ is inductively defined on terms as
follows:
\begin{itemize}
\item
$x \pto{P}_\RR x$ if $x$ is a variable and $P = \varnothing$.
\item
$\ell\sigma \pto{P}_\RR r\sigma$ 
if $\ell \to r$ is an $\RR$-rule, $\sigma$ is a substitution, and
$P = \{ \epsilon \}$.
\item
$f(s_1,\dots,s_n) \pto{P}_\RR f(t_1,\dots,t_n)$ if
$f$ is an $n$-ary function symbol in $\FF$,
$s_i \pto{P_i}_\RR t_i$ holds for all $1 \leqslant i \leqslant n$, and
\(
P = \{ i \cdot p \mid
\text{$1 \leqslant i \leqslant n$ and $p \in P_i$} \}
\).
\end{itemize}
We write $s \pto{}_\RR t$ if $s \pto{P}_\RR t$ for some 
set $P$ of positions.
\end{defi}

Note that $\pto{}_\RR$ is reflexive and the inclusions
${\to_\RR} \subseteq {\pto{}_\RR} \subseteq {\to_\RR^*}$ hold.  As the
latter entails ${\to_\RR^*} = {\pto{}_\RR^*}$, we obtain the following
useful characterizations.

\begin{lem}
\label{lem:tait-martin-loef}
A TRS $\RR$ is confluent if and only if $\pto{}_\RR$ is confluent.
Similarly, TRSs $\RR$ and $\SS$ commute if and only if $\pto{}_\RR$ and
$\pto{}_\SS$ commute.
\end{lem}

\section{Parallel Closedness}
\label{sec:pc}

Toyama made two variations of Huet's parallel closedness theorem~\cite{H80}
in 1981~\cite{T81} and in 1988~\cite{T88}, but their relation has not been
known.  In this section we recall his and related results, and then
show that Toyama's earlier result subsumes the later one.  
For brevity we omit the subscript $\RR$ from 
$\to_\RR$, $\pto{}_\RR$, and $\cp{\RR}$
when it is clear from the contexts.

\begin{defiC}[\cite{H80}]
A TRS is \emph{parallel closed} if
${\cp{}} \subseteq {\pto{}}$ holds.
\end{defiC}
\begin{thmC}[\cite{H80}]
\label{thm:H80}
A left-linear TRS is confluent if it is parallel closed.
\end{thmC}

In 1988, Toyama showed that the closing form for \emph{overlay} critical
pairs, originating from root overlaps, can be relaxed.  
We write
$t \cpT{}{> \epsilon} u$ if $t \cpT{}{p} u$ holds
for some $p > \epsilon$.

\begin{defiC}[\cite{T88}]
A TRS is \emph{almost parallel closed} if
\(
{\cpT{}{\epsilon}} \subseteq {\pto{} \cdot \fromT{*}}
\)
and
\(
{\cpT{}{> \epsilon}} \subseteq {\pto{}}
\)
hold.
\end{defiC}
\begin{thmC}[\cite{T88}]
\label{thm:T88}
A left-linear TRS is confluent if it is almost parallel closed.
\end{thmC}

\begin{exa}
\label{ex:T88}
Consider the following left-linear and non-terminating TRS, 
which is a variant of the TRS in \cite[Example~5.4]{G96}.
\begin{align*}
\m{a}(x) & \to \m{b}(x)
&
\m{f}(\m{a}(x),\m{a}(y)) & \to \m{g}(\m{f}(\m{a}(x),\m{a}(y)))
\\
\m{f}(\m{b}(x),y) & \to \m{g}(\m{f}(\m{a}(x),y))
&
\m{f}(x,\m{b}(y)) & \to \m{g}(\m{f}(x,\m{a}(y)))
\end{align*}
Out of the three critical pairs,
two critical pairs including the next diagram
(i) are closed by single parallel steps. The remaining pair (ii)
joins by performing a single parallel step on each side:
\begin{center}
\begin{tabular}{c@{\qquad}c}
\begin{tikzpicture}[D,baseline=(t)]
\node (s) at (0,2) {$\m{f}(\m{a}(x),\m{a}(y))$};
\node (t) at (0,0.0) {$\m{f}(\m{b}(x),\m{a}(y))$};
\node (u) at (3.3,2) {$\m{g}(\m{f}(\m{a}(x),\m{a}(y)))$};
\draw[->]
 (s) edge node[left]  {$1$} (t)
 (s) edge node[above] {$\epsilon$} (u)
 (t) edge[dashed,bend right,-||->] (u)
;
\end{tikzpicture}
&
\begin{tikzpicture}[D,baseline=(t)]
\node (s) at (0,2) {$\m{f}(\m{b}(x),\m{b}(y))$};
\node (t) at (0,0)   {$\m{g}(\m{f}(\m{a}(x),\m{b}(y)))$};
\node (u) at (3.3,2) {$\m{g}(\m{f}(\m{b}(x),\m{a}(y)))$};
\node (v) at (3.3,0)   {$\m{g}(\m{f}(\m{b}(x),\m{b}(y)))$};
\draw[->]
 (s) edge node[left] {$\epsilon$} (t)
 (s) edge node[above] {$\epsilon$} (u)
 (t) edge[dashed,-||->] (v)
 (u) edge[dashed,-||->] (v)
;
\end{tikzpicture}
\\
(i) & (ii)
\end{tabular}
\end{center}
Thus, the TRS is almost parallel closed.  Hence, the TRS is confluent.
\end{exa}

Inspired by almost parallel closedness, Gramlich~\cite{G96} developed
a confluence criterion based on \emph{parallel critical pairs} in 1996.
Let $t$ be a term and let $P$ be a set of parallel positions in $t$.
We write $\Var(t,P)$ for the union of $\Var(t|_p)$ for all $p \in P$.
By $t[u_p]_{p \in P}$ we denote the term that results from replacing in
$t$ the subterm at $p$ by a term $u_p$ for all $p \in P$.

\begin{defi}
Let $\RR$ and $\SS$ be TRSs, 
$\ell \to r$ a variant of an $\SS$-rule, and
$\{ \ell_p \to r_p \}_{p \in P}$ a family of variants of
$\RR$-rules, where $P$ is a set of positions.
A local peak
\[
(\ell\sigma)[r_p\sigma]_{p \in P} \pfromB{}{\RR} \ell\sigma \rto_\SS r\sigma
\]
is called a \emph{parallel critical peak} between $\RR$ and $\SS$ if
the following conditions hold:
\begin{itemize}
\item
$P \subseteq \Pos_\FF(\ell)$ is a non-empty set of parallel
positions in $\ell$,
\item
none of rules $\ell \to r$ and $\ell_p \to r_p$ for $p \in P$
shares a variable with other rules,
\item
$\sigma$ is a most general unifier of
$\{ \ell_p \approx (\ell|_p) \}_{p \in P}$, and
\item
if $P = \{\epsilon\}$ then
$\ell_\epsilon \to r_\epsilon$ is not a variant of $\ell \to r$.
\end{itemize}
When $t \pfromB{P}{\RR} s \rto_\SS u$ is a parallel critical peak, the
pair $(t,u)$ is called a \emph{parallel critical pair}, and denoted by
\(
t \mathrel{{\pfromB{P}{\RR}}{\rtimes}{\rto_\SS}} u
\).
In the case of $P \nsubseteq \{ \epsilon \}$
the parallel critical pair is written as
$t \mathrel{{\pfromB{>\epsilon}{\RR}}{\rtimes}{\rto_\SS}} u$.
Whenever no confusion arises, we abbreviate
${\pfromB{}{\RR}}{\rtimes}{\rto_\RR}$ to
${\pfrom{}}{\rtimes}{\rto}$.
\end{defi}

Consider a local peak $t \pfromB{P}{\RR} s \rto_\SS u$ that
employs a rule $\ell_p \to r_p$ at $p \in P$ in the left step and a rule
$\ell \to r$ in the right step.  We say that the peak is \emph{orthogonal}
if either $P \cap \Pos_\FF(\ell) = \varnothing$, or $P = \{\epsilon\}$ and
$\ell_\epsilon \to r_\epsilon$ is a variant of $\ell \to r$.\footnote{%
As the name suggests, every local peak $\pfromB{P}{\RR} \cdot \rto_\RR$ is
orthogonal for \emph{orthogonal TRSs}, see \secref{sec:orthogonality}.}
A local peak $t \mathrel{_\RR{\xleftarrow{\smash{p}}}} s \rto_\SS u$ is
\emph{orthogonal} if $t \pfromB{\{p\}}{\RR} s \rto_\SS u$ is.

\begin{thmC}[\cite{G96}]
\label{thm:G96}
A left-linear TRS
is confluent if the inclusions
${\cp{}} \subseteq {\pto{} \cdot \fromT{*}}$ and
\(
{\pfrom{>\epsilon}}%
{\rtimes}%
{\rto} 
\subseteq {\to^*}
\)
hold.
\end{thmC}

Unfortunately, this criterion by Gramlich does not subsume (almost)
parallel closedness.

\begin{exa}[Continued from Example~\ref{ex:T88}]
\label{ex:G96}
The TRS admits the parallel critical peak\linebreak
\\[-10pt]  % Prevent text from inline from extending into line above
\(
\m{f}(\m{b}(x),\m{b}(y)) \pfrom{\smash{\{1,2\}}} \m{f}(\m{a}(x),\m{a}(y))
\xrightarrow{\smash{\epsilon}} \m{g}(\m{f}(\m{a}(x),\m{a}(y)))
\).
However,
\(
\m{f}(\m{b}(x),\m{b}(y)) \to^* \m{g}(\m{f}(\m{a}(x),\m{a}(y)))
\)
does not hold.
\end{exa}

As noted in the paper \cite{G96}, Toyama~\cite{T81} had already obtained
in 1981 a closedness result that subsumes Theorem~\ref{thm:G96}.
His idea is to impose variable conditions on parallel steps $\pto{}$.

\begin{thmC}[\cite{T81}]
\label{thm:T81}
A left-linear TRS is confluent if the following conditions hold:
\begin{enumerate}[(a)]
\item
\label{strongly_parallel_closed_1}
The inclusion ${\cp{}} \subseteq {\pto{}^{} \cdot \fromT{*}}$ holds.
\item
\label{strongly_parallel_closed_2}
For every parallel critical peak $t \pfrom{\smash{P}} s \rto u$ there
exist a term $v$ and a set $P'$ of parallel positions such that
$t \to^* v \pfrom{\smash{P'}} u$ and
$\Var(v,P') \subseteq \Var(s,P)$.
\end{enumerate}
\end{thmC}

\begin{exa}[Continued from Example~\ref{ex:G96}]
The confluence of the TRS in \cref{ex:T88} can be shown by \cref{thm:T81}.
Since condition \pref{strongly_parallel_closed_1} of \cref{thm:T81} follows
from the almost parallel closedness, it is enough to verify condition
\pref{strongly_parallel_closed_2}.  
The following parallel critical peak, which \cref{thm:G96} fails to handle,
admits the following diagram:
\begin{center}
\begin{tikzpicture}[D, baseline=(s)]
\node (s) at (0.0,2.0) {$\m{f}(\m{a}(x),\m{a}(y))$};
\node (t) at (0.0,0.0) {$\m{f}(\m{b}(x),\m{b}(y))$};
\node (u) at (4.0,2.0) {$\m{g}(\m{f}(\m{a}(x),\m{a}(y)))$};
\node (v) at (4.0,0.0) {$\m{g}(\m{f}(\m{a}(x),\m{b}(y)))$};
\draw[->]
 (s) edge[-||->] node[left=1mm] {$\{1,2\}$} (t)
 (s) edge node[above] {$\epsilon$} (u)
;
\draw[dashed,->]
 (t) edge node[below] {\phantom{$\RR$}} (v)
 (u) edge[-||->] node[right=1mm] {$\{1 \cdot 2\}$} (v)
;
\end{tikzpicture}
\end{center}
Because 
\(
\Var(\m{g}(\m{f}(\m{a}(x),\m{b}(y))), \{1 \cdot 2\})
= \{ y \} \subseteq \{x,y\} =
\Var(\m{f}(\m{a}(x),\m{a}(y)),\{1,2\})
\)
holds, the parallel critical peak satisfies
condition~\pref{strongly_parallel_closed_2} in
\cref{thm:T81}.
Similarly, we can find
suitable diagrams for the other parallel critical peaks.  Hence,
\pref{strongly_parallel_closed_2} holds for the TRS.
\end{exa}

Now we show that Theorem~\ref{thm:T81} even subsumes Theorem~\ref{thm:T88}.
The first part of the next lemma is a strengthened version of the
Parallel Moves Lemma~\cite[Lemma~6.4.4]{BN98}. Here a variable
condition like \cref{thm:T81} is associated.  The second part of the lemma
is irrelevant here but will be used in the subsequent sections.  Note that
the second part corresponds to \cite[Lemma~55]{ZFM15}.
We write $\sigma \pto{}_\RR \tau$ if $x\sigma \pto{}_\RR x\tau$ for all
variables $x$.

\begin{lem}\label{lem:PML}
Let $\RR$ be a TRS and $\ell \to r$ a left-linear rule.
Consider a local peak $\Gamma$ of the form
$t \pfromB{P}{\RR} s \rto_{\{\ell \to r\}} u$.
\begin{enumerate}[(a)]
\item
\label{PML_orthogonal_case}
If $\Gamma$ is orthogonal,
$t \rto_{\{\ell \to r\}}^= v \pfromB{P'}{\RR} u$ and 
$\Var(v,P') \subseteq \Var(s,P)$ for some $v$ and $P'$.
\item
\label{PML_overlapping_case}
Otherwise, there exist a parallel critical peak 
$t_0 \pfromB{P_0}{\RR} s_0 \rto_{\{\ell \to r\}} u_0$
and substitutions $\sigma$ and $\tau$ such that
$s = s_0\sigma$, $t = t_0\tau$, $u = u_0\sigma$, 
$\sigma \pto{}_\RR \tau$,
$t_0\sigma \pto{P \setminus P_0}_\RR t_0\tau$, and
$P_0 \subseteq P$.
\end{enumerate}
See the diagrams in \cref{fig:PML}.
\end{lem}
\begin{figure}[b]
\centering
\begin{tabular}{c@{\hspace{4em}}c}
\begin{tikzpicture}[D]
\node (s) at (0,4) {$s$};
\node (t) at (0,0) {$t$};
\node (u) at (4,4) {$u$};
\node (v) at (4,0) {$v$};
\draw[->]
 (s) edge[-||->] node[left,xshift=-1mm] {$P$} (t)
 (s) edge node[above] {$\epsilon$} (u)
;
\draw[->,dashed]
 (t) edge
     node[near end,xshift=2mm,yshift=0.5em] {$=$}
     node[above] {$\epsilon$} (v)
 (u) edge[-||->] node[right,xshift=1mm] {$P'$} (v)
;
\end{tikzpicture}
&
\begin{tikzpicture}[D]
\node (s)       at (0,4)   {$s$};
\node           at (0.4,4) {$=$};
\node (s0sigma) at (1,4)   {$s_0\sigma$};
\node (t0sigma) at (1,2)   {$t_0\sigma$};
\node (t0tau)   at (1,0)   {$t_0\tau$};
\node           at (0.4,0) {$=$};
\node (t)       at (0,0)   {$t$};
\node (u)       at (5,4)   {$u$};
\node           at (4.6,4) {$=$};
\node (u0sigma) at (4,4)   {$u_0\sigma$};
\draw[->]
 (s) edge[-||->] node[left,xshift=-1mm] {$P$} (t)
 (s0sigma) edge node[above] {$\epsilon$} (u0sigma)
 (s0sigma) edge[-||->,dashed] node[right,xshift=1mm] {$P_0$} (t0sigma)
 (t0sigma) edge[-||->,dashed] node[right,xshift=1mm] {$P \setminus P_0$} (t0tau)
;
\end{tikzpicture}
\\
(a) $\Gamma$ is orthogonal  & (b) $\Gamma$ is not orthogonal
\end{tabular}
\caption{The claims of \cref{lem:PML}.}
\label{fig:PML}
\end{figure}

\begin{proof}
\begin{enumerate}[(a)]
\item
Suppose that $\Gamma$ is orthogonal.
If $s \pto{\{\epsilon\}}_{\{\ell' \to r'\}} t$ holds for some variant
$\ell' \to r'$
of $\ell \to r$ then $t = u$. Thus, $t \to^= u \pfrom{\varnothing} u$.
Otherwise, $P \cap \Pos_\FF(\ell) = \varnothing$.
Since $s \rto_{\{ \ell \to r \}} u$
holds, there exists a substitution $\sigma$ with $s = \ell\sigma$ and
$u = r\sigma$.  As $\ell\sigma \pto{P} t$, $\ell$ is linear, and
$P \cap \Pos_\FF(\ell) = \varnothing$, straightforward induction on 
$\ell$ shows existence of $\tau$ such that $t = \ell\tau$ and
$\sigma \pto{}_{\RR} \tau$.
Take $v = r\tau$ and define $P'$ as follows:
\[
P' = \{ p'_1 \cdot p_2 \mid \text{$p_1 \cdot p_2 \in P$,
$p'_1 \in \Pos_\VV(r)$, and
$\ell|_{p_1} = r|_{p'_1}$ for some $p_1 \in \Pos_\VV(\ell)$} \}
\]
Clearly, $t \rto_{\{\ell \to r\}} v$ holds.  So it remains to show 
$u \pto{P'}_\RR v$ and $\Var(v,P') \subseteq \Var(s,P)$.  
Let $p'$ be an arbitrary position in $P'$.  There exist positions 
$p_1 \in \Pos_\VV(\ell)$, $p'_1 \in \Pos_\VV(r)$, and $p_2$ such that 
$p' = p'_1 \cdot p_2$, $p_1 \cdot p_2 \in P$, and
$\ell|_{p_1} = r|_{p'_1}$.
Denoting $p_1 \cdot p_2$ by $p$, we have the identities:
\begin{alignat*}{6}
u|_{p'}
& = (r\sigma)|_{p'_1 \cdot p_2}
&& = (r|_{p'_1}\sigma)|_{p_2}
&& = (\ell|_{p_1}\sigma)|_{p_2}
&& = (\ell\sigma)|_{p_1 \cdot p_2}
&& = s|_p
\\
v|_{p'}
& = (r\tau)|_{p'_1 \cdot p_2}
&& = (r|_{p'_1}\tau)|_{p_2}
&& = (\ell|_{p_1}\tau)|_{p_2}
&& = (\ell\tau)|_{p_1 \cdot p_2}
&& = t|_p
\end{alignat*}
From $s \pto{P}_\RR t$ we obtain
$s|_p \rto_\RR t|_p$ and thus
$u|_{p'} \rto_\RR v|_{p'}$.  Therefore, $u \pto{P'}_\RR v$ is obtained.
Moreover, we have
\(
\Var(v|_{p'}) = \Var(t|_p)
\subseteq \Var(s|_p)
\subseteq \Var(s,P)
\).
As $\Var(v,P')$ is the union of $\Var(v|_{p'})$ for all $p' \in P'$, the
desired inclusion $\Var(v, P') \subseteq \Var(s,P)$ follows.
\item
Suppose that $\Gamma$ is not orthogonal.
By $\ell_p \to r_p$ we denote the rule employed at the rewrite
position $p \in P$ in $s \pto{P}_\RR t$.
Let $P_0 = P \cap \Pos_\FF(\ell)$ and $P_1 = P \setminus P_0$.
Since $P$ is a set of parallel positions, $s \pto{P} t$ is split
into the two steps $s \pto{P_0}_{\RR} v \pto{P_1}_{\RR} t$,
where $v = s[t|_p]_{p \in P_0}$.

First, we show that $v \pfrom{P_0} s \rto_{\{\ell \to r\}} u$ is
an instance of a parallel critical peak. Let $p$ be an arbitrary position
in $P_0$.  
Because of $s \rto_{\{\ell \to r\}} u$, we have $s = \ell\mu$ and 
$u = r\mu$ for some $\mu$. 
Suppose that $\ell'_p \to r'_p$ is a renamed variant of $\ell_p \to r_p$
with fresh variables.  There exists a substitution $\mu_p$ such that 
$s|_p = \ell'_p\mu_p$ and $t|_p = r'_p\mu_p$.
Note that $\Dom(\mu) \cap \Dom(\mu_p) = \varnothing$.
We define the substitution $\nu$ as follows:
\[
\nu(x) =
\begin{cases}
x\mu_p & \text{if $p \in P_0$ and $x \in \Var(\ell'_p)$} \\
x\mu   & \text{otherwise}
\end{cases}
\]
Because every $\ell'_p$ with $p \in P_0$ is linear and do not share
variables with each other,
$\nu$ is well-defined.  Since $\ell$
neither share variables with $\ell'_p$,
we obtain the identities:
\[
\ell'_p\nu = \ell'_p\mu_p = s|_p = \ell|_p\mu = \ell|_p\nu 
\]
Thus, $\nu$ is a unifier of $E = \{ \ell'_p \approx \ell|_p \}_{p \in P_0}$.
Let $V$ denote the set of all variables occurring in $E$.
According to \cite[Proposition~4.10]{E85}, there exists a most general unifier $\nu'$ of
$E$ such that $\Dom(\nu') \subseteq V$.
Thus, there is a substitution $\sigma$ with $\nu = \nu'\sigma$.  Let $s_0 = \ell\nu'$, 
$t_0 = (\ell\nu')[r'_p\nu']_{p \in P_0}$, and $u_0 = r\nu'$.
The peak $t_0 \pfrom{P_0} s_0 \rto u_0$ is a parallel critical peak, and
$v \pfrom{P_0} s \rto u$ is an instance of the peak by
the substitution $\sigma$:
\begin{align*}
s_0\sigma & = \ell\nu'\sigma = \ell\nu = \ell\mu = s
\\
t_0\sigma
& = (\ell\nu'\sigma)[r'_p\nu'\sigma]_{p \in P_0}
  = (\ell\nu)[r'_p\nu]_{p \in P_0}
  = (\ell\mu)[r'_p\mu_p]_{p \in P_0}
  = v
\\
u_0\sigma & = r\nu'\sigma = r\nu = r\mu = u
\end{align*}

Next, we construct a substitution $\tau$ so that it satisfies 
$\sigma \pto{}_\RR \tau$ and $t_0\sigma \pto{P_1}_\RR t_0\tau$.
Given a variable $x \in \Var(\ell)$, we write $p_x$ for a variable
occurrence of $x$ in $\ell$.  Due to linearity of $\ell$, the position
$p_x$ is uniquely determined.
Let $W = \Var(\ell) \setminus \Var(\ell, P_0)$.
Note that $W \cap V = \varnothing$ holds.
We define the substitution $\tau$ as follows:
\[
\tau(x) =
\begin{cases}
t|_{p_x}
& \text{if $x \in W$} \\
x\sigma & \text{otherwise}
\end{cases}
\]
To verify $\sigma \pto{}_\RR \tau$, consider an arbitrary variable
$x$.
We show $x\sigma \pto{}_\RR x\tau$.
If $x \notin W$ then 
$x\sigma = x\tau$, from which the claim follows.
Otherwise, 
the definitions of $V$ and $\nu'$ yield the implications:
\[
x \in W
\implies x \notin V
\implies x \notin \Dom(\nu')
\implies x\nu' = x
\]
So $s_0|_{p_x} = x$ follows from the identities:
\[
s_0|_{p_x} = (\ell\nu)|_{p_x} = \ell|_{p_x}\nu = x \nu = x
\]
Let $Q_x = \{ q \mid p_xq \in P_1 \}$.
As $s \pto{P_0}_{\RR} v \pto{P_1}_{\RR} t$ implies
$s|_{p_x} = v|_{p_x} \pto{Q_x}_{\RR} t|_{p_x}$, 
we obtain
\(
x\sigma = s_0|_{p_x}\sigma = (s_0\sigma)|_{p_x} = s|_{p_x}
\pto{Q_x}_{\RR} t|_{p_x} = x\tau
\).
Therefore, the claim is verified.

The remaining task is to show $t_0\sigma \pto{P_1}_{\RR} t_0\tau$. 
Let $p \in P_1$.
As $s_0|_{p_x} = x$ and $s_0 \pto{P_0}_{\RR} t_0$ imply $x = t_0|_{p_x}$,
the equation $(s_0\sigma)|_p = (t_0\sigma)|_p$ 
follows. By the definition of $\tau$ we have
$(t_0\tau)|_{p_x} = t|_{p_x}$, which leads to
$(t_0\tau)|_p = t|_p$.  Hence, we obtain the relations
\[
(t_0\sigma)|_{p} = (s_0\sigma)|_p = s|_p
\pto{\{\epsilon\}}_{\RR}
t|_p = (t_0\tau)|_{p}
\]
which entails the desired parallel step $t_0\sigma \pto{P_1}_\RR t_0\tau$.
\qedhere
\end{enumerate}
\end{proof}

For almost parallel closed TRSs the above statement is extended to 
local peaks $\pfrom{} \cdot \pto{}$ of parallel steps.
In its proof we measure 
parallel steps $s \pto{P} t$ in such a local peak by the
\emph{total size of contractums} $|t|_P$, namely
the sum of $|(t|_p)|$ for all $p \in P$.
Note that this measure attributes to \cite{OO97,LJ14}.

\begin{lem}\label{lem:almost_to_t81}
Consider a left-linear almost parallel closed TRS.  If
$t \pfrom{P_1} s \pto{P_2} u$ then
\begin{itemize}
\item
$t \to^* v_1 \pfrom{P'_1} u$ for some $v_1$ and $P'_1$
with $\Var(v_1,P'_1) \subseteq \Var(s,P_1)$, and
\item
$t \pto{P'_2} v_2 \fromT{*} u$ for some $v_2$ and $P'_2$
with $\Var(v_2,P'_2) \subseteq \Var(s,P_2)$.
\end{itemize}
\end{lem}
\begin{proof}
Let $\Gamma\colon t \pfrom{P_1} s \pto{P_2} u$ be a local peak.
We show the claim by well-founded induction on
$(|t|_{P_1}+|u|_{P_2},s)$ with respect to $\succ$.
Here $(m,s) \succ (n,t)$ if either $m > n$, or $m = n$ and $t$
is a proper subterm of $s$.
Depending on the shape of $\Gamma$, we distinguish six cases.
\begin{enumerate}
\item
If $P_1$ or $P_2$ is empty
then the claim follows from the
fact: $\Var(v,P) \subseteq \Var(w,P)$
if $w \pto{P} v$.

\item
If $P_1$ or $P_2$ is $\{ \epsilon \}$ and $\Gamma$ is orthogonal then
Lemma~\ref{lem:PML}\pref{PML_orthogonal_case} applies.

\item
If $P_1 = P_2 = \{\epsilon\}$ and $\Gamma$ is not orthogonal
then $\Gamma$ is an instance of a critical peak.
By almost parallel closedness $t \to^* v_1 \pfrom{Q_1} u$ and
$t \pto{Q_2} v_2 \fromT{*} u$ for some $v_1$, $v_2$, $Q_1$, and $Q_2$.
For each $k \in \{1,2\}$ we have $s \to^* v_k$, so
$\Var(v_k) \subseteq \Var(s)$ follows.
Thus,
\(
\Var(v_k, Q_k)
\subseteq \Var(v_k)
\subseteq \Var(s)
= \Var(s,\{\epsilon\})
\).
The claim holds.

\item
If $P_1 {{}\nsubseteq{}} \{\epsilon\}$, $P_2 = \{ \epsilon \}$, and $\Gamma$ is not
orthogonal then there is $p \in P_1$ such that
$s' \xleftarrow{\smash{p}} s \rto u$ is an instance of a critical peak and
$s' \pto{P_1 \setminus \{p\}} t$ 
follows by \cref{lem:PML}\pref{PML_overlapping_case} where $P = \{p\}$.
By the almost parallel closedness $s' \pto{\smash{P_2'}} u$ for some
$P_2'$.  Since $P_2'$ is a set of parallel positions in $u$, we have
$|u|_{\{\epsilon\}} = |u| \geqslant |u|_{P_2'}$.
As $|u|_{\{\epsilon\}} \geqslant |u|_{P_2'}$ and
$|t|_{P_1} > |t|_{P_1 \setminus \{p\}}$ yield
$|t|_{P_1} + |u|_{\{\epsilon\}} > |t|_{P_1 \setminus \{p\}} + |u|_{P_2'}$,
we obtain the inequality:
\[
(|t|_{P_1} + |u|_{P_2}, s) \succ
(|t|_{P_1 \setminus \{ p \}} + |u|_{P_2'}, s')
\]
Thus, the claim follows by the induction hypothesis for
$t \pfrom{\smash{P_1 \setminus \{p\}}} s' \pto{\smash{P_2'}} u$ and the inclusions
\(
\Var(s', P_1 \setminus \{p\}) \subseteq \Var(s,P_1)
\)
and
\(
\Var(s', P_2') \subseteq \Var(s,\{\epsilon\})
\).

\item
If $P_1 = \{\epsilon\}$, $P_2 \nsubseteq \{ \epsilon \}$,
and $\Gamma$ is not orthogonal then
the proof is analogous to the last case.

\item
If $P_1 \nsubseteq \{\epsilon\}$ and $P_2 \nsubseteq \{\epsilon\}$ then
we may assume $s = f(\seq{s})$, $t = f(\seq{t})$, $u = f(\seq{u})$, and
$t_i \pfrom{P_1^i} s_i \pto{P_2^i} u_i$ for all
$1 \leqslant i \leqslant n$.
Here $P_k^i$ denotes the set $\{ p \mid i \cdot p \in P_k \}$.
For each $i \in \{1,\ldots,n\}$, we have $|t|_{P_1} \geqslant |t_i|_{P_1^i}$ and
$|u|_{P_2} \geqslant |u_i|_{P^i_2}$, and therefore
$|t|_{P_1} + |u|_{P_2} \geqslant |t_i|_{P^i_1} + |u_i|_{P^i_2}$.
So we deduce the following inequality:
\[
(|t|_{P_1} + |u|_{P_2}, s) \succ
(|t_i|_{P_1^i} + |u_i|_{P_2^i}, s_i)
\]
Consider the $i$-th peak
\(
t_i \pfrom{\smash{P_1^i}} s_i \pto{\smash{P_2^i}} u_i
\).
By the induction hypothesis it admits valleys of the forms
$t_i \to^* v_1^i \pfrom{Q_1^i} u_i$
and
$t_i \pto{Q_2^i} v_2^i \fromT{*} u_i$
such that
$\Var(v_k^i, Q_k^i) \subseteq \Var(s_i, P_k^i)$ for both
$k \in \{1,2\}$.  For each $k$, define
\(
Q_k = \{ i \cdot q \mid \text{$1 \leqslant i \leqslant n$ and
$q \in Q_k^i$} \}
\)
and $v_k = f(v_k^1,\ldots,v_k^n)$.
Then we have
$t \to^* v_1 \pfrom{Q_1} u$
and
$t \pto{Q_2} v_2 \fromT{*} u$.
Moreover,
\[
\Var(v_k, Q_k) =
\bigcup_{i=1}^n \Var(v_k^i, Q_k^i) \subseteq
\bigcup_{i=1}^n \Var(s_i, P_k^i) = \Var(s, P_k)
\]
holds.
Hence, the claim follows.
\qedhere
\end{enumerate}
\end{proof}

\begin{thm}
\label{thm:T81_subsumes_T88}
Every left-linear and almost parallel closed TRS satisfies
conditions~\pref{strongly_parallel_closed_1} and
\pref{strongly_parallel_closed_2} of Theorem~\ref{thm:T81}.
In other words, Theorem~\ref{thm:T81} subsumes Theorem~\ref{thm:T88}.
\end{thm}
\begin{proof}
Since (parallel) critical peaks are instances of $\pfrom{} \cdot \pto{}$,
Lemma~\ref{lem:almost_to_t81} entails the claim. 
\end{proof}

Note that Theorem~\ref{thm:T88} does not subsume Theorem~\ref{thm:T81}
as witnessed by the TRS consisting of the four rules
$\m{f(a)} \to \m{c}$,
$\m{a} \to \m{b}$,
$\m{f(b)} \to \m{b}$, and
$\m{c} \to \m{b}$.
In \secref{sec:prlc} we will see that \cref{thm:T81} is subsumed by a
variant of rule labeling.

\section{Decreasing Diagrams with Commuting Subsystems}
\label{sec:ddc}

We make a variant of decreasing diagrams~\cite{vO94,vO08},
which will be used in the subsequent sections for deriving compositional
confluence criteria for term rewrite systems.
First we recall the commutation version of the technique~\cite{vO08}.
Let $\AA = (A, \{ \to_{1,\alpha} \}_{\alpha \in I})$ and
$\BB = (A, \{ \to_{2,\beta} \}_{\beta \in J})$ be
$I$-indexed and $J$-indexed ARSs on the same domain, respectively.
Let $>$ be a well-founded order on $I \cup J$.  By $\curlyvee\alpha$ we
denote the set
$\{ \beta \in {I \cup J} \mid \alpha > \beta \}$, and by
$\curlyvee\alpha\beta$ we denote $(\curlyvee\alpha) \cup (\curlyvee\beta)$.
We say that a local peak $b \fromB{1,\alpha} a \to_{2,\beta} c$ is
\emph{decreasing} if
\[
b \xleftrightarrow[\curlyvee \alpha]{*}
\cdot \xrightarrow[2,\beta]{=}
\cdot \xleftrightarrow[\curlyvee\alpha\beta]{*}
\cdot \xleftarrow[1,\alpha]{=}
\cdot \xleftrightarrow[\curlyvee\beta]{*} c
\]
holds.  Here $\fromto_K$ stands for the union of $\fromB{1,\gamma}$ and
$\to_{2,\gamma}$ for all $\gamma \in K$.  The ARSs $\AA$ and $\BB$ are
\emph{decreasing} if 
every local peak $b \fromB{1,\alpha} a \to_{2,\beta} c$ with
$(\alpha, \beta) \in I \times J$ is decreasing.
In the case of $\AA = \BB$, we simply say that $\AA$ is decreasing.

\begin{thmC}[\cite{vO08}]
\label{thm:dd}
If two ARSs are decreasing then they commute.
\end{thmC}

We present the abstract principle of our compositional criteria.  The idea 
of using the least index in the decreasing diagram technique
is taken from \cite{JL12,FvO13,DFJL22}.

\begin{thm} \label{thm:ddc}
Let $\AA = (A, \{\to_{1,\alpha}\}_{\alpha \in I})$ 
and $\BB = (A, \{\to_{2,\beta}\}_{\beta \in I})$
be $I$-indexed ARSs equipped with a well-founded order $>$ on $I$.
Suppose that $\bot$ is the least element in $I$ and
$\to_{1,\bot}$ and $\to_{2,\bot}$ commute.
The ARSs $\AA$ and $\BB$ commute if 
every local peak $\fromB{1,\alpha} \cdot \to_{2,\beta}$ with
$(\alpha,\beta) \in I^2 \setminus \{(\bot,\bot)\}$ is decreasing.
\end{thm}
\begin{proof}
We define the two ARSs $\AA' = (A, \{ \To_{1,\alpha} \}_{\alpha \in I})$
and $\BB' = (A, \{ \To_{2,\alpha} \}_{\alpha \in I})$ as follows:
\[
{\To_{i,\alpha}}  =
\begin{cases}
{\to_{i,\alpha}^*} & \text{if $\alpha = \bot$} \\
{\to_{i,\alpha}}   & \text{otherwise}
\end{cases}
\]
Since
${\to_\AA^*} = {\To_\AA^*}$ and ${\to_\BB^*} = {\To_\BB^*}$, the
commutation of $\AA$ and $\BB$ follows from that of $\AA'$ and $\BB'$.
We show the latter by proving decreasingness of $\AA'$ and $\BB'$ with
respect to the given well-founded order $>$.  Let $\Gamma$ be a local peak of
form $\FromB{1,\alpha} \cdot \To_{2,\beta}$.  We distinguish four cases.
\begin{itemize}
\item
If neither $\alpha$ nor $\beta$ is $\bot$ then
decreasingness of $\Gamma$ follows from 
the assumption.
\item
If both $\alpha$ and $\beta$ are $\bot$ then the commutation of
$\to_{1,\bot}$ and $\to_{2,\bot}$ yields the inclusion:
\[ 
{\xLeftarrow[1,\bot]{} \cdot \xRightarrow[2,\bot]{}}
\subseteq
{\xRightarrow[2,\bot]{} \cdot \xLeftarrow[1,\bot]{}}
\]
Thus $\Gamma$ is decreasing.
\item
If $\beta > \alpha = \bot$ then 
we have
\(
{\fromB{1,\alpha} \cdot \to_{2,\beta}}
\subseteq
{\to_{2,\beta}^= \cdot \fromto_{\curlyvee\beta}^*}
\)
Therefore, easy induction on $n$ shows the inclusion
\(
{\fromBT{1,\alpha}{n} \cdot \to_{2,\beta}}
\subseteq
{\to_{2,\beta}^= \cdot \fromto_{\curlyvee\beta}^*}
\)
for all $n \in \NN$.  Thus,
\[
{\xLeftarrow[1,\alpha]{} \cdot \xRightarrow[2,\beta]{}}
\;=\;
{\xleftarrow[1,\alpha]{*} \cdot \xrightarrow[2,\beta]{}}
\;\subseteq\;
{\xrightarrow[2,\beta]{=} \cdot \xleftrightarrow[\curlyvee\beta]{*}}
\;=\;
{\xRightarrow[2,\beta]{=} \cdot \xLeftrightarrow[\curlyvee\beta]{*}}
\]
holds, where $\Fromto_J$ stands for $\FromB{1,J} \cup \To_{2,J}$.
Hence $\Gamma$ is decreasing.
\item
The case that $\alpha > \beta = \bot$ is analogous to the
last case.
\qedhere
\end{itemize}
\end{proof}

\section{Orthogonality}
\label{sec:orthogonality}

As a first example of compositional confluence criteria 
for term rewrite systems, we pick up 
a compositional version of Rosen's confluence criterion by orthogonality~\cite{R73}.
\emph{Orthogonal} TRSs are left-linear TRSs having no critical pairs.
Their confluence property can be shown by decreasingness of parallel steps.
We briefly recall its proof.  Left-linear TRSs are \emph{mutually
orthogonal} if ${\fromB{\RR}}{\rtimes}{\rto_\SS} = \varnothing$ and
${\fromB{\SS}}{\rtimes}{\rto_\RR} = \varnothing$. Note that 
orthogonality of $\RR$ and mutual orthogonality of $\RR$ and $\RR$ are
equivalent.

\begin{lemC}[{\cite[Theorem~9.3.11]{BN98}}]
\label{lem:commuting diamond}
For mutually orthogonal TRSs $\RR$ and $\SS$ the inclusion
\(
{\pfromB{}{\RR} \cdot \pto{}_\SS} \subseteq 
{{\pto{}_{\SS} \cdot \pfromB{}{\RR}}}
\)
holds.
\end{lemC}
 
\begin{thmC}[\cite{R73}]
\label{thm:R73}
Every orthogonal TRS $\RR$ is confluent.
\end{thmC}
\begin{proof}
Let $\AA = (\TT(\FF,\VV), \{ \pto{}_1 \})$ be the ARS equipped with the empty
order $>$ on $\{ 1 \}$, where $\pto{}_1 = {\pto{}_\RR}$.  
According to \cref{lem:tait-martin-loef} and \cref{thm:dd}, it is
enough to show that $\AA$ is decreasing.
Since \cref{lem:commuting diamond} yields
\(
{\pfromB{}{1} \cdot \pto{}_1} \subseteq {\pto{}_1 \cdot \pfromB{}{1}}
\),
the decreasingness of $\AA$ follows.
\end{proof}
 
The theorem can be recast as a compositional criterion that uses a
confluent subsystem $\CC$ of a given TRS $\RR$.  For this sake we switch
the underlying criterion from \cref{thm:dd} to \cref{thm:ddc}, setting the
relation of the least index $\bot$ to $\pto{}_\CC$. 
 
\begin{thm}
\label{thm:tentative}
A left-linear TRS $\RR$ is confluent if 
$\RR$ and $\RR \setminus \CC$ are mutually orthogonal
for some confluent TRS $\CC$ with $\CC \subseteq \RR$.
\end{thm}
\begin{proof}
Suppose that $\CC \subseteq \RR$ and $\CC$ is confluent.
Let $\AA = (\TT(\FF,\VV),\{ \pto{}_0, \pto{}_1 \})$ be the ARS equipped
with the well-founded order $1 > 0$, where ${\pto{}_0} = {\pto{}_\CC}$ and
${\pto{}_1} = {\pto{}_{\RR \setminus \CC}}$.  
Since $\CC$ is confluent, $\CC$ and $\CC$ commute. So $\pto{}_0$ and
$\pto{}_0$ commute too.
According to \cref{lem:tait-martin-loef} and \cref{thm:ddc}, it is 
sufficient
to show that all local peak $\pfromB{}{i} \cdot \pto{}_j$ with
$(i,j) \neq (0,0)$ are decreasing.
Since $\RR$ and $\RR \setminus \CC$ are
mutually orthogonal,
$\RR \setminus \CC$ and $\RR \setminus \CC$ as well as $\CC$ and 
$\RR \setminus \CC$ are mutually orthogonal.  Therefore,
\cref{lem:commuting diamond} yields
the following inclusions:
\begin{align*}
{\pfromB{}{\RR \setminus \CC} \cdot \pto{}_{\RR \setminus \CC}} \subseteq
{\pto{}_{\RR \setminus \CC} \cdot \pfromB{}{\RR \setminus \CC}}
&&
{\pfromB{}{\CC} \cdot \pto{}_{\RR \setminus \CC}} \subseteq
{\pto{}_{\RR \setminus \CC} \cdot \pfromB{}{\CC}}
\end{align*}
So
\(
{\pfromB{}{k} \cdot \pto{}_m} \subseteq {\pto{}_m \cdot \pfromB{}{k}}
\)
holds for all $(k,m) \in \{0,1\}^2 \setminus \{(0,0)\}$, from which
the decreasingness of $\AA$ follows.  Hence, \cref{thm:ddc} applies.
\end{proof}

We can derive a more general criterion by exploiting the flexible valley
form of decreasing diagrams.  We will adopt parallel critical pairs.
It causes no loss of confluence proving power of \cref{thm:tentative} as
${\pfromB{}{\RR}}{\rtimes}{\rto_\SS} = \varnothing$
is equivalent to
${\fromB{\RR}}{\rtimes}{\rto_\SS} = \varnothing$.

\begin{thm}
\label{thm:SH22a}
A left-linear TRS $\RR$ is confluent if 
${\pcp{\RR}} \subseteq {\fromto_\CC^*}$ holds for
some confluent TRS $\CC$ with $\CC \subseteq \RR$.
\end{thm}
\begin{proof}
Recall the ARS used in the proof of \cref{thm:tentative}.  According to
\cref{lem:tait-martin-loef} and \cref{thm:ddc}, it is sufficient to
show that every local peak
\[
\Gamma : t \pfrom[k]{P} s \pto[m]{Q} u
\]
with $(k, m) \neq (0,0)$ is decreasing.  To this end, we show $t \pto{}_m
\cdot \pfromto{}_0^* \cdot \pfromB{}{k} u$ by structural induction on $s$.
Depending on the shape of $\Gamma$, we distinguish five cases.
\begin{enumerate}
\item 
If $P$ or $Q$ is empty then the claim is trivial.
\item
If $P$ or $Q$ is $\{ \epsilon \}$ and $\Gamma$ is orthogonal then
\cref{lem:PML}\pref{PML_orthogonal_case} yields
$t \pto{}_m \cdot \pfromB{}{k} u$.
\item \label{SH22a_fig_case}
If $P \neq \varnothing$, $Q = \{ \epsilon \}$, and $\Gamma$ is not
orthogonal then by \cref{lem:PML}\pref{PML_overlapping_case} there exist a
parallel critical peak $t_0 \pfromB{}{k} s_0 \rto_m u_0$ and
substitutions $\sigma$ and $\tau$ such that $s = s_0\sigma$, $t = t_0\tau$, 
$u = u_0\sigma$, and $\sigma \pto{}_k \tau$.
The assumption $t_0 \fromto_\CC^* u_0$ yields
$t_0\tau \pfromto{}_0^* u_0\tau$
because $\to$ is closed under 
substitutions and ${\to} \subseteq {\pto{}}$.
Therefore, 
\(
t = t_0\tau
\pfromto{}_0^*
u_0\tau
\pfromB{}{k}
u_0\sigma = u
\)
follows.
\item
If $P = \{ \epsilon \}$, $Q \neq \varnothing$, and $\Gamma$ is not
orthogonal then the proof is analogous to the last case.
\item
If $P \nsubseteq \{ \epsilon \}$ and $Q \nsubseteq \{\epsilon\}$ then
$s$, $t$, and $u$ can be written as
$f(\seq{s})$,
$f(\seq{t})$, and
$f(\seq{u})$ respectively, and moreover,
$t_i \pfromB{}{k} s_i \pto{}_m u_i$
holds for all $1 \leqslant i \leqslant n$.  For every $i$ 
the induction hypothesis yields
\(
t_i \pto{}_m v_i \pfromto{}_0^* w_i \pfromB{}{k} u_i
\)
for some $v_i$ and $w_i$.
Therefore, the desired conversion 
\(
t 
\pto{}_m v \pfromto{}_0^* w \pfromB{}{k}
u
\)
holds for $v = f(\seq{v})$ and $w = f(\seq{w})$.
\qedhere
\end{enumerate}
\end{proof}
\begin{figure}[t]
\centering
\begin{tikzpicture}[D]
\node (s)        at (0.0,4.0) {$s_0\sigma$};
\node (t0sigma)  at (0.0,2.0) {$t_0\sigma$};
\node (t)        at (0.0,0.0) {$\makebox[0mm][r]{$t ={}$} t_0\tau$};
\node (u)        at (4.5,4.0) {$u_0\sigma \makebox[0mm][l]{${}= u$}$};
\node (u0tau)    at (4.5,0.0) {$u_0\tau$};
\draw[->]
 (s)
 edge[-||->]
 node[right=1mm] {$k$}
 (t0sigma)
 (t0sigma)
 edge[-||->]
 node[right=1mm] {$k$}
 (t)
 (s)
 edge[-||->]
 node[above=0.5mm] {$\{\epsilon\}$}
 node[below=0.5mm] {$m$}
 (u)
 (u)
 edge[-||->,dashed]
 node[right=1mm] {$k$}
 (u0tau)
 (t)
 edge[<->,dashed,-||->]
 node[above,xshift=3em] {$*$}
 node[below=1mm] {$0$}
 (u0tau)
;
\end{tikzpicture}
\caption{Proof of Theorem~\ref{thm:SH22a} (\ref{SH22a_fig_case}).}
\label{fig:SH22a}
\end{figure}

From Takahashi's proposition~\cite{T93}
(see also \cite[Proposition~9.3.5]{TeReSe}) we can deduce that
${\pcp{\RR}} \subseteq {=}$ is equivalent to ${\cp{\RR}} \subseteq {=}$.
Thus, \cref{thm:SH22a} subsumes \cref{thm:tentative}.
Note that when $\CC = \varnothing$, \cref{thm:SH22a} 
simulates the weak orthogonality criterion.

\begin{exa}
\label{ex:SH22a}
By successive application
of \cref{thm:SH22a} we show the confluence of the left-linear TRS
$\RR$ (COPS~\cite{HNM18} number \texttt{62}), taken from~\cite{OO03}:
\begin{alignat*}{6}
1\colon~ && x - \m{0} & \to x 
&
7\colon~ && \m{gcd}(x, \m{0}) & \to x 
&
13\colon~ && \m{if}(\m{true},  x, y) & \to x 
\\
2\colon~ && \m{0} - x & \to \m{0} 
&
8\colon~ && \m{gcd}(\m{0}, x) & \to x 
&
14\colon~ && \m{if}(\m{false}, x, y) & \to y
\\
3\colon~ && \m{s}(x) - \m{s}(y) & \to x - y \quad
&
9\colon~ && \m{gcd}(x, y) & \to \m{gcd}(y, \m{mod}(x, y))
\\
4\colon~ && x < \m{0} & \to \m{false} 
&
10\colon~ && \m{mod}(x, \m{0}) & \to x 
\\
5\colon~ && \m{0} < \m{s}(y) & \to \m{true} 
&
11\colon~ && \m{mod}(\m{0}, y) & \to \m{0} 
\\
6\colon~ && \m{s}(x) < \m{s}(y) & \to x < y 
&
12\colon~ && \m{mod}(x, \m{s}(y)) & \to
\makebox[0mm][l]{$\m{if}(x <  \m{s}(y), x, \m{mod}(x - \m{s}(y), \m{s}(y)))$} 
\end{alignat*}
Let $\CC = \{5,7,8,10,11,13\}$. 
The six non-trivial parallel critical pairs of $\RR$ are
\begin{align*}
(x,\m{gcd}(\m{0},\m{mod}(x,\m{0}))) &&
(y,\m{gcd}(y,\m{mod}(\m{0},y))) &&
(\m{0}, \m{if}(\m{0} < \m{s}(y),\m{0}, \m{mod}(\m{0} - \m{s}(y), \m{s}(y))))
\end{align*}
and their symmetric versions.  All of them are joinable by $\CC$.  So it
remains to show that $\CC$ is confluent.  Because $\CC$ only admits trivial
parallel critical pairs, ${\pcp{\CC}} \subseteq {\fromto_\varnothing^*}$
holds.  Therefore, the confluence of $\CC$ is concluded if we show the
confluence of the empty system.  The latter claim is trivial. This completes
the proof.
\end{exa}

\cref{thm:SH22a} is a generalization of Toyama's unpublished result:

\begin{corC}[\cite{T17}]
\label{cor:T17}
A left-linear TRS $\RR$ is confluent if 
${\pcp{\RR}} \subseteq {\fromto_\CC^*}$ holds for
some terminating and confluent TRS $\CC$ with $\CC \subseteq \RR$.
\end{corC}

\section{Rule Labeling}
\label{sec:prlc}

In this section we recast the \emph{rule labeling}
criterion~\cite{vO08,ZFM15,DFJL22} in a compositional form.
Rule labeling is a direct application of decreasing
diagrams to confluence proofs for TRSs.  It labels rewrite steps by
their employed rewrite rules and compares indexes of them.
Among others, we focus on
the variant of rule labeling based on
parallel critical pairs, introduced by Zankl et al.~\cite{ZFM15}.

\begin{defi}
Let $\RR$ be a TRS.  A \emph{labeling function} for $\RR$ is a function
from $\RR$ to $\NN$.
Given a labeling function $\phi$ and a number
$k \in \NN$, we define the TRS $\RR_{\phi,k}$ as follows:
\[
\RR_{\phi,k} = \{ \ell \to r \in \RR \mid \phi(\ell \to r) \leqslant k \}
\]
The relations $\to_{\RR_{\phi,k}}$ and $\pto{}_{\RR_{\phi,k}}$ are
abbreviated to $\to_{\phi,k}$ and $\pto{}_{\phi,k}$.
Let $\phi$ and $\psi$ be labeling functions for $\RR$.
We say that a local peak 
\(
t \pfrom[\phi,k]{P} s \xrightarrow[\psi,m]{\epsilon} u
\)
is 
\emph{$(\psi,\phi)$-decreasing}
if
\[
t \xleftrightarrow[\curlyvee k]{*}
\cdot \pto[\psi,m]{}
\cdot \xleftrightarrow[\curlyvee km]{*} 
v \pfrom[\phi,k]{\smash{P'}}
\cdot \xleftrightarrow[\curlyvee m]{*} u
\]
and $\Var(v,P') \subseteq \Var(s,P)$ for some set $P'$ of parallel
positions and term $v$.
Here $\fromto_K$ stands for the union of $\fromB{\phi,k}$ and
$\to_{\psi,k}$ for all $k \in K$.
\end{defi}

The following theorem is a variant of the rule labeling method based
on parallel critical pairs.

\begin{thmC}[{\cite[Theorem~56]{ZFM15}}]
\label{thm:ZFM15}
Let $\RR$ be a left-linear TRS, and $\phi$ and $\psi$ its labeling
functions. The TRS $\RR$ is confluent if
the following conditions hold for all $k,m \in \NN$. 
\begin{itemize}
\item
Every parallel critical peak of form
\(
t \pfrom[\phi,k]{} s \xrightarrow[\psi,m]{\epsilon} u
\)
is $(\psi,\phi)$-decreasing.

\item
Every parallel critical peak of form
\(
t \pfrom[\psi,m]{} s \xrightarrow[\phi,k]{\epsilon} u
\)
is $(\phi,\psi)$-decreasing.
\end{itemize}
\end{thmC}

With a small example we illustrate the usage of rule labeling.

\begin{exa}
Consider the 
left-linear TRS $\RR$:
\begin{align*}
 (x + y) + z & \to x + (y + z)
&
 x + (y + z) & \to (x + y) + z
\end{align*}
We define the labeling functions $\phi$ and $\psi$ as follows:
$\phi(\ell \to r) = 0$ and $\psi(\ell \to r) = 1$ for all $\ell \to r \in \RR$.
All parallel critical peaks can be closed by
$\to_{\phi,0}$-steps, like the following diagram:
\begin{center}
\begin{tikzpicture}[D]
\node (s) at (0.0,2) {$\makebox[0mm][r]{$s ={}$}((x + y) + z) + w$};
\node (t) at (0.0,0) {$(x + (y + z)) + w$};
\node (u) at (8.0,2) {$(x + y) + (z + w)$};
\node (w) at (8.0,0) {$(x + y) + (z + w) \makebox[0mm][l]{${}= v$}$};
\node (v) at (4.0,0) {$((x + y) + z) + w$};
\draw[->]
  (s) edge[-||->] node[left=1mm] {$\{1\}$} node[right=1mm,Label] {$\phi,0$} (t)
  (s) edge node[above] {$\epsilon$} node[below,Label] {$\psi,1$} (u)
;
\draw[dashed,->]
  (u) edge[-||->] node[left=1mm] {$\varnothing$} node[right=1mm,Label] {$\phi,0$} (w)
  (w) edge node[below,Label] {$\phi,0$} (v)
  (v) edge node[below,Label] {$\phi,0$} (t)
;
\end{tikzpicture}
\end{center}
As
\(
\Var(v, \varnothing)
= \varnothing
\subseteq \{x,y,z\} 
= \Var(s, \{1\})
\),
this parallel critical peak is $(\psi,\phi)$-decreasing.
In a similar way the other peaks can also be verified.
Hence, the TRS $\RR$ is confluent.
\end{exa}

We make the rule labeling compositional.  The following lemma is used
for composing parallel steps.

\begin{lemC}[{\cite[Lemma~51(b)]{ZFM15}}]
\label{lem:recomposition}
If $s \pto{P}_\RR t$, $\sigma \pto{}_\RR \tau$, and $x\sigma = x\tau$ for
all $x \in \Var(t,P)$ then $s\sigma \pto{}_\RR t\tau$.
\end{lemC}

The next theorem is 
a compositional version of the rule labeling criterion.
Note that by taking $\CC := \RR_{\phi,0} = \RR_{\psi,0}$ it can be used as a
compositional confluence criterion parameterized by $\CC$.

\begin{thm} \label{thm:prlc}
Let $\RR$ be a left-linear TRS, and $\phi$ and $\psi$ its labeling
functions.  Suppose that $\RR_{\phi,0}$ and $\RR_{\psi,0}$ commute.  The
TRS $\RR$ is confluent if the following conditions hold for all
$(k,m) \in \NN^2 \setminus \{(0,0)\}$.
\begin{itemize}
\item
Every parallel critical peak of form
\(
t \pfrom[\phi,k]{} s \xrightarrow[\psi,m]{\epsilon} u
\)
is $(\psi,\phi)$-decreasing.
\item
Every parallel critical peak of form
\(
t \pfrom[\psi,m]{} s \xrightarrow[\phi,k]{\epsilon} u
\)
is $(\phi,\psi)$-decreasing.
\end{itemize}
\end{thm}
\begin{proof}
Consider the ARSs
$(\TT(\FF,\VV),\{\pto{}_{\phi,k}\}_{k \in \NN})$
and 
$(\TT(\FF,\VV),\{\pto{}_{\psi,m}\}_{m \in \NN})$.
According to \cref{lem:tait-martin-loef} and \cref{thm:ddc},
it is sufficient to show that every local peak 
\[
\Gamma\colon t \pfrom[\phi,k]{P} s \pto[\psi,m]{Q} u
\]
with $(k, m) \neq (0,0)$ is decreasing.  To this end, we perform structural
induction on $s$.  Depending on the shape of $\Gamma$, we distinguish
five cases.
\begin{enumerate}
\item
If $P$ or $Q$ is empty then the claim is trivial.

\item
If $P$ or $Q$ is $\{ \epsilon \}$ and $\Gamma$ is orthogonal then
\cref{lem:PML}\pref{PML_orthogonal_case} yields
$t \pto[\psi,m]{} \cdot \pfrom[\phi,k]{} u$.

\item \label{prlc_fig_case}
If $P \neq \varnothing$, $Q = \{ \epsilon \}$, and $\Gamma$ is not
orthogonal then by \cref{lem:PML}\pref{PML_overlapping_case} there
exist a parallel critical peak
$t_0 \pfrom[\phi,k']{P_1} s_0 \xrightarrow[\psi,m]{\epsilon} u_0$
and substitutions $\sigma$ and $\tau$ such that
$k' \leqslant k$,
$t = t_0\tau$,
$u = u_0\sigma$, 
$\sigma \pto[\phi,k]{} \tau$,
$t_0\sigma \pto[\phi,k]{P \setminus P_1} t_0\tau$, and
$P_1 \subseteq P$.
We distinguish two subcases.\footnote{
The preliminary version of this paper~\cite{SH22} lacks
this case analysis.}
If $k' = 0$ and $m = 0$ then
$t_0 \pfromto[0]{*} u_0$.  As $\pto{}$ is closed under substitutions,
$t_0\tau \pfromto[0]{*} u_0\tau$ follows.  The step can be written as
$t_0\tau \pfromto[\curlyvee k]{*} u_0\tau$ because $(k,m) \neq (0,0)$ and
$m = 0$ imply $k > 0$. Summing them up, we obtain the sequence
\[
t = t_0\tau \pfromto[\curlyvee k]{*} u_0\tau \pfrom[\phi,k]{} u_0\sigma = u
\]
from which we conclude decreasingness of $\Gamma$.
Otherwise, $k' > 0$ or  $m > 0$ holds.
The assumption yields
\[
t_0
\pfromto[\curlyvee k']{*}
\cdot \pto[\psi,m]{}
\cdot \pfromto[\curlyvee k'm]{*}
v_0 \pfrom[\phi,k']{P'_1}
w_0 \pfromto[\curlyvee m]{*}
u_0
\]
and $\Var(v_0,P'_1) \subseteq \Var(s_0,P_1)$ for some
$v_0$, $w_0$, and $P'_1$.
Since $k' \leqslant k$ and the rewrite steps are closed under substitutions,
the following relations are obtained:
\begin{align*}
t_0\tau & \pfromto[\curlyvee k]{*}
\cdot \pto[\psi,m]{}
\cdot \pfromto[\curlyvee km]{*}
v_0\tau 
&
w_0\sigma & \pfromto[\curlyvee m]{*}
u_0\sigma
\end{align*}
Since $t_0\sigma|_p = t_0\tau|_p$ holds for all $p \in P_1$,
the identity $x\sigma = x\tau$ holds for all 
$x \in \Var(s_0,P_1)$.
Therefore, $x\sigma = x\tau$ holds for all
$x \in \Var(v_0,P'_1)$.  
Because $w_0 \pto[\phi,k]{\smash{P'_1}} v_0$, $\sigma \pto[\phi,k]{} \tau$,
and $x\sigma = x\tau$ for all $x \in \Var(v_0,P'_1)$ hold,
\cref{lem:recomposition} yields
$w_0\sigma \pto[\phi,k]{} v_0\tau$.  Hence,
the decreasingness of $\Gamma$ is witnessed by the following
sequence:
\[
t = t_0\tau
\pfromto[\curlyvee k]{*}
\cdot \pto[\psi,m]{}
\cdot \pfromto[\curlyvee km]{*}
v_0\tau \pfrom[\phi,k]{}
w_0\sigma \pfromto[\curlyvee m]{*}
u_0\sigma = u
\]
Note that the construction is depicted in \cref{fig:prlc}.

\item
If $P = \{ \epsilon \}$, $Q \neq \varnothing$,
and $\Gamma$ is not orthogonal
then the proof is analogous to the last case.
\item
If $P \nsubseteq \{ \epsilon \}$ and $Q \nsubseteq \{\epsilon\}$ then
$s$, $t$, and $u$ can be written as
$f(\seq{s})$,
$f(\seq{t})$, and
$f(\seq{u})$ respectively, and moreover,
\(
t_i \pfrom[\phi,k]{} s_i \pto[\psi,m]{} u_i
\)
holds for all $1 \leqslant i \leqslant n$.
By the induction hypotheses we have
\(
t_i
\pfromto[\curlyvee k]{*}
\cdot \pto[\psi,m]{}
\cdot \pfromto[\curlyvee km]{*}
\cdot \pfrom[\phi,k]{}
\cdot \pfromto[\curlyvee m]{*}
u_i
\)
for all $1 \leqslant i \leqslant n$.  Therefore, we obtain the desired
relations:
\[
t = f(\seq{t})
\pfromto[\curlyvee k]{*}
\cdot \pto[\psi,m]{}
\cdot \pfromto[\curlyvee km]{*}
\cdot \pfrom[\phi,k]{}
\cdot \pfromto[\curlyvee m]{*}
f(\seq{u}) = u
\]
Hence $\Gamma$ is decreasing.
\qedhere
\end{enumerate}
\end{proof}
\begin{figure}[t]
\centering
\begin{tikzpicture}[D]
\node (s)        at (0.0,4.5) {$s_0\sigma$};
\node (t0sigma)  at (0.0,1.5) {$t_0\sigma$};
\node (t)        at (0.0,0.0) {$\makebox[0mm][r]{$t ={}$} t_0\tau$};
\node (u)        at (4.5,4.5) {$u_0\sigma \makebox[0mm][l]{${}= u$}$};
\node (w0sigma1) at (1.5,1.5) {$\cdot$};
\node (w0sigma2) at (3.0,1.5) {$\cdot$};
\node (w0sigma)  at (4.5,1.5) {$v_0\sigma$};
\node (v0sigma)  at (4.5,3.0) {$w_0\sigma$};
\node (w0tau1)   at (1.5,0)   {$\cdot$};
\node (w0tau2)   at (3.0,0.0) {$\cdot$};
\node (w0tau)    at (4.5,0.0) {$v_0\tau$};
\draw[->]
 (s)
 edge[-||->]
 node[left=1mm] {$P_1$}
 node[right=1mm,Label] {$\phi,k$}
 (t0sigma)
 (t0sigma)
 edge[-||->]
 node[right=1mm,Label] {$\phi,k$}
 (t)
 (s)
 edge[-||->]
 node[above=0.5mm] {$\{\epsilon\}$}
 node[below=0.5mm,Label] {$\psi,m$}
 (u)
 (t0sigma)
 edge[<->,dashed,-||->]
 node[anchor=south,xshift=1em] {$*$}
 node[anchor=north,Label] {$\curlyvee k$}
 (w0sigma1)
 (w0sigma1)
 edge[-||->,dashed]
 node[below,Label] {$\psi,m$}
 (w0sigma2)
 (w0sigma2)
 edge[<->,dashed,-||->]
 node[anchor=south,xshift=-1em] {$*$}
 node[anchor=north,Label,yshift=-1pt] {$\curlyvee km$}    % Keeping text out of graphic
 (w0sigma)
 (u)
 edge[<->,dashed,-||->]
 node[anchor=east,yshift=-1em] {$*$}
 node[anchor=west,xshift=1mm,Label] {$\curlyvee m$}
 (v0sigma)
 (v0sigma)
 edge[-||->,dashed]
 node[left=1mm] {$P'_1$}
 node[right=1mm,Label] {$\phi,k$}
 (w0sigma)
 (w0sigma)
 edge[-||->,dashed]
 node[right=1mm,Label] {$\phi,k$}
 (w0tau)
 (t)
 edge[<->,dashed,-||->]
 node[above,xshift=1em] {$*$}
 node[below,Label] {$\curlyvee k$}
 (w0tau1)
 (w0tau1)
 edge[-||->,dashed]
 node[below,Label] {$\psi,m$}
 (w0tau2)
 (w0tau2)
 edge[<->,dashed,-||->]
 node[above,xshift=-1em] {$*$}
 node[below,Label,yshift=-1pt] {$\curlyvee km$}
 (w0tau)
 (v0sigma)
 edge[-||->,bend left=75,dashed]
 node[right=1mm,Label] {$\phi,k$}
 (w0tau)
;
\end{tikzpicture}
\caption{Proof of Theorem~\ref{thm:prlc}(\ref{prlc_fig_case}).}
\label{fig:prlc}
\end{figure}

The original version of rule labeling (\cref{thm:ZFM15})
is a special case of \cref{thm:prlc}:  
Suppose that labeling functions $\phi$ and $\psi$ for a left-linear TRS
$\RR$ satisfy the conditions of \cref{thm:ZFM15}.  By taking
the labeling functions $\phi'$ and $\psi'$ with
\begin{align*}
\phi'(\ell \to r) & = \phi(\ell \to r) + 1
&
\psi'(\ell \to r) & = \psi(\ell \to r) + 1
\end{align*}
\cref{thm:prlc} applies for $\phi'$, $\psi'$, and the empty TRS $\CC$.

The next example shows the combination of our rule labeling
variant (\cref{thm:prlc}) with Knuth--Bendix' criterion (\cref{thm:KB70}).

\begin{exa}
\label{ex:prlc}
Consider the left-linear TRS $\RR$:
\begin{align*}
1\colon\; \m{0} + x & \to x
&
2\colon\;  (x + y) + z & \to x + (y + z)
&
3\colon\; x + (y + z) & \to (x + y) + z
\end{align*}
Let $\CC = \{1,2\}$.  We define the labeling functions $\phi$ and $\psi$ as
follows:
\[
\phi(\ell \to r) = \psi(\ell \to r) =
\begin{cases}
0 & \text{if $\ell \to r \in \CC$} \\
1 & \text{otherwise}
\end{cases}
\]
For instance, the parallel critical pairs involving
rule $3$ admit the following diagrams:
\begin{center}
\begin{tikzpicture}[D]
\node (s) at (0.0,2) {$x + (\m{0} + z)$};
\node (t) at (0.0,0) {$x + z$};
\node (u) at (3.0,2) {$(x + \m{0}) + z$};
\node (v) at (3.0,0) {$x + (\m{0} + z)$};
\draw[->]
  (s) edge[-||->] node[left=1mm] {$\{2\}$} node[right=1mm,Label] {$\phi,0$} (t)
  (s) edge node[above] {$\epsilon$} node[below,Label] {$\psi,1$} (u)
;
\draw[->,dashed]
  (u) edge node[right,Label] {$\phi,0$} (v)
  (v) edge node[below,Label] {$\phi,0$} (t)
;
\end{tikzpicture}
\hfil
\begin{tikzpicture}[D]
\node (s) at (0.0,2) {$x + (y + (z + w))$};
\node (t) at (0.0,0) {$x + ((y + z) + w)$};
\node (u) at (4.5,2) {$(x + y) + (z + w)$};
\node (v) at (4.5,0) {$x + (y + (z + w))$};
\draw[->]
  (s) edge[-||->] node[left=1mm] {$\{2\}$} node[right=1mm,Label] {$\phi,1$} (t)
  (s) edge[->] node[above] {$\epsilon$} node[below,Label] {$\psi,1$} (u)
;
\draw[->,dashed]
  (u) edge node[right,Label] {$\phi,0$} (v)
  (v) edge[-||->] node[below,Label] {$\phi,1$} node[above=0.5mm] {$\{2\}$} (t)
;
\end{tikzpicture}
\end{center}
They fit for the conditions of \cref{thm:prlc}.  The other
parallel critical pairs also admit suitable diagrams.  Therefore, it remains
to show that $\CC$ is confluent.  Since $\CC$ is terminating and all its
critical pairs are joinable, confluence of $\CC$ follows by Knuth and
Bendix' criterion (\cref{thm:KB70}).  Thus, $\RR_{\phi,0}$
and $\RR_{\psi,0}$ commute because $\RR_{\phi,0} = \RR_{\psi,0} = \CC$.
Hence, by \cref{thm:prlc} we conclude that $\RR$ is confluent.
\end{exa}

While a proof for \cref{thm:SH22a} is given in \secref{sec:orthogonality},
here we present an alternative proof based on \cref{thm:prlc}.

\begin{proof}[Proof of \cref{thm:SH22a}]
Define the labeling functions $\phi$ and $\psi$ as in \cref{ex:prlc}.
Then \cref{thm:prlc} applies.
\end{proof}

Unlike \cref{thm:SH22a}, successive applications of \cref{thm:prlc} are
not more powerful than a single application of it. To see it, suppose that
confluence of a left-linear finite TRS $\RR$ is shown by
\cref{thm:prlc} with labeling functions $\phi_\RR$ and $\psi_\RR$, where
confluence of the employed subsystem $\CC$ is shown by the theorem with
$\phi_\CC$, $\psi_\CC$, and a confluent subsystem $\CC'$.  The confluence
of $\RR$ can be shown by \cref{thm:prlc} with the confluent subsystem
$\CC'$ and the labeling functions $\phi$ and $\psi$:
\begin{align*}
\phi(\ell \to r) & =
\begin{cases}
\phi_\CC(\ell \to r) & \text{if $\ell \to r \in \CC$} \\
\phi_\RR(\ell \to r) + m & \text{otherwise}
\end{cases}
&
\psi(\ell \to r) & =
\begin{cases}
\psi_\CC(\ell \to r) & \text{if $\ell \to r \in \CC$} \\
\psi_\RR(\ell \to r) + m & \text{otherwise}
\end{cases}
\end{align*}
Here
\(
m = \max(\{ 0 \} \cup \{ \phi_\CC(\ell \to r), \psi_\CC(\ell \to r) \mid
\ell \to r \in \CC \})
\).
As a consequence, whenever confluence is shown by successive application of
\cref{thm:prlc}, it can also be shown by the original theorem
(\cref{thm:ZFM15}).

We conclude the section by stating
that rule labeling based on parallel
critical pairs (\cref{thm:ZFM15})
subsumes parallel closedness based on parallel
critical pairs (\cref{thm:T81}):
Suppose that conditions (a) and (b) of \cref{thm:T81}
hold.
We define $\phi$ and $\psi$ as the constant rule
labeling functions $\phi(\ell \to r) = 1$ and $\psi(\ell \to r) = 0$.
By using structural induction as well as
Lemmata~\ref{lem:PML} and~\ref{lem:recomposition}
we can prove the implication
\[
t \pfrom[\phi,1]{P_1} s \pto[\psi,0]{} u
\implies
\text{$t \xrightarrow[\psi,0]{*} v \pfrom[\phi,1]{P'_1} u$
and $\Var(v,P'_1) \subseteq \Var(s,P_1)$ for some $P'_1$}
\]
Thus, the conditions of \cref{thm:ZFM15} follow.  As a consequence, our
compositional version (\cref{thm:prlc}) is also a generalization of
parallel closedness.

\section{Critical Pair Systems}
\label{sec:cpsc}

The last example of compositional criteria in this paper is a variant
of the confluence criterion by critical pair systems~\cite{HM11}.
It is known that the original criterion is a generalization of the orthogonal criterion
(\cref{thm:R73}) and Knuth and Bendix' criterion (\cref{thm:KB70}) for
left-linear TRSs.

\begin{defi}
The \emph{critical pair system} $\CPS(\RR)$ of a TRS $\RR$ is defined as
the TRS:
\[
\{ s \to t, s \to u \mid
\text{$t \fromB{\RR} s \rto_\RR u$ is a critical peak} \}
\]
\end{defi}

\begin{thmC}[\cite{HM11}]
\label{thm:HM11}
A left-linear and locally confluent TRS $\RR$ is confluent if
$\CPS(\RR)/\RR$ is terminating (i.e., $\CPS(\RR)$ is relatively
terminating with respect to $\RR$).
\end{thmC}

The theorem is shown by using the decreasing
diagram technique (\cref{thm:dd}), see~\cite{HM11}.

\begin{exa}
Consider the left-linear and non-terminating TRS $\RR$:
\begin{align*}
\m{s}(\m{p}(x)) & \to \m{p}(\m{s}(x))
&
\m{p}(\m{s}(x)) & \to x
&
\infty & \to \m{s}(\infty)
\end{align*}
The TRS $\RR$ admits two critical pairs and they are joinable:
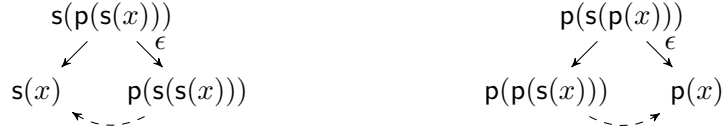
\begin{center}
\begin{tikzpicture}[D]
\node (s) at (1,1) {$\m{s}(\m{p}(\m{s}(x)))$};
\node (t) at (0,0) {$\m{s}(x)$};
\node (u) at (2,0) {$\m{p}(\m{s}(\m{s}(x)))$};
\draw[->]
 (s) edge (t)
 (s) edge node[above right=-1mm] {$\epsilon$} (u)
;
\draw[->,dashed]
 (u) edge[bend left] (t)
;
\end{tikzpicture}
\hfil
\begin{tikzpicture}[D]
\node (s) at (1,1) {$\m{p}(\m{s}(\m{p}(x)))$};
\node (t) at (0,0) {$\m{p}(\m{p}(\m{s}(x)))$};
\node (u) at (2,0) {$\m{p}(x)$};
\draw[->]
 (s) edge (t)
 (s) edge node[above right=-1mm] {$\epsilon$}(u)
;
\draw[->,dashed]
 (t) edge[bend right] (u)
;
\end{tikzpicture}
\end{center}
The critical pair system $\CPS(\RR)$ consists of the four rules:
\begin{align*}
\m{s}(\m{p}(\m{s}(x))) & \to \m{s}(x)
&
\m{p}(\m{s}(\m{p}(x))) & \to \m{p}(\m{p}(\m{s}(x)))
\\
\m{s}(\m{p}(\m{s}(x))) & \to \m{p}(\m{s}(\m{s}(x)))
&
\m{p}(\m{s}(\m{p}(x))) & \to \m{p}(x)
\end{align*}
The termination of $\CPS(\RR)/\RR$ can be shown by, e.g., the termination tool
\NaTT (cf.\xspace \secref{sec:experiments}).
Hence the confluence of $\RR$ follows by \cref{thm:HM11}.
\end{exa}

We argue about the parallel critical pair version of $\CPS(\RR)$:
\[
\PCPS(\RR) = \{ s \to t, s \to u \mid
\text{$t \pfromB{}{\RR} s \rto_\RR u$ is a parallel critical peak} \}
\]
Interestingly, replacing $\CPS(\RR)$ by $\PCPS(\RR)$ 
in \cref{thm:HM11} results in the same criterion (see~\cite{ZFM15}).
Since
\(
{\to_{\CPS(\RR)}}
\subseteq {\to_{\PCPS(\RR)}}
\subseteq {\to_{\CPS(\RR)} \cdot \pto{}_\RR}
\)
holds,
${\to_{\CPS(\RR)/\RR}} = {\to_{\PCPS(\RR)/\RR}}$
follows.
So the termination of $\PCPS(\RR)/\RR$ is equivalent to that of
$\CPS(\RR)/\RR$. 
However, a compositional form of \cref{thm:HM11} may benefit from the use of
parallel critical pairs, as seen in \secref{sec:orthogonality}.

\begin{defi}
Let $\RR$ and $\CC$ be TRSs. The \emph{parallel critical pair system}
$\PCPS(\RR,\CC)$ of $\RR$ modulo $\CC$ is defined as the TRS:
\[
\{ s \to t, s \to u \mid
\text{$t \pfromB{}{\RR} s \rto_\RR u$ is a parallel critical peak but not
$t \fromto^*_\CC u$}
\}
\]
\end{defi}

Note that $\PCPS(\RR,\varnothing) \subseteq \PCPS(\RR)$ holds in general,
and $\PCPS(\RR,\varnothing) \subsetneq \PCPS(\RR)$ when $\RR$ admits a
trivial critical pair.

The next lemma relates $\PCPS(\RR,\CC)$ to closing forms of parallel
critical peaks.

\begin{lem}
\label{lem:key}
Let $\RR$ be a left-linear TRS and $\RR_1$, $\RR_2$, and $\CC$ 
subsets of $\RR$, and let $\PP = \PCPS(\RR,\CC)$.  Suppose that
${\pcp{\RR}} \subseteq {\to_\RR^* \cdot \fromBT{\RR}{*}}$ holds.
If $t \pfromB{}{\RR_1} s \pto{}_{\RR_2} u$ then
\begin{enumerate}[label=(\roman*)]
\item
\label{lem:key:join}
$t \pto{}_{\RR_2} \cdot \fromto_\CC^* \cdot \pfromB{}{\RR_1} u$, or
\item
\label{lem:key:peak}
$t \pfromB{}{\RR_1} t' \fromB{\PP} s \to_\PP u' \pto{}_{\RR_2} u$
and $t' \to_\RR^* \cdot \fromBT{\RR}{*} u'$ for some $t'$ and $u'$.
\end{enumerate}
\end{lem}
\begin{proof}
Let $\Gamma\colon t \pfromB{P}{\RR_1} s \pto{Q}_{\RR_2} u$ be a local peak.
We use structural induction on $s$. Depending on the form of $\Gamma$, we
distinguish five cases.
\begin{enumerate}
\item
If $P$ or $Q$ is the empty set then \pref{lem:key:join} holds trivially.
\item
If $P$ or $Q$ is $\{ \epsilon \}$
and $\Gamma$ is orthogonal then 
\pref{lem:key:join} follows by \cref{lem:PML}\pref{PML_orthogonal_case}.
\item
If $P \neq \varnothing$, $Q = \{ \epsilon \}$, and $\Gamma$ is not orthogonal then 
we distinguish two cases.
\begin{itemize}
\item
If there exist $P_0$, $t_0$, $u_0$, and $\sigma$ such that
``$P_0 \subseteq P$,
\(
t \pfromB{}{\RR_1}
t_0\sigma \pfromB{\smash{P_0}}{\RR_1}
s \rto_{\RR_2} u_0\sigma = u
\),
and $t_0 \pcp{\RR} u_0$'' but not $t_0 \fromto_\CC^* u_0$.  Take
$t' = t_0\sigma$ and $u' = u_0\sigma$.  Then 
\(
t_0\tau \pfromB{}{\RR_1} t_0\sigma \fromB{\PP} s \to_\PP u_0\sigma = u
\)
holds and by the assumption $t' \to_\RR^* \cdot \fromBT{\RR}{*} u'$
also holds.  Hence \pref{lem:key:peak} follows.
\item
Otherwise, whenever $P_0$, $t_0$, $u_0$, and $\sigma$ satisfy the
conditions quoted in the last item,
$t_0 \fromto_\CC^* u_0$ holds.
Because $\Gamma$ is not orthogonal,
by \cref{lem:PML}\pref{PML_overlapping_case} there exist
$P_0$, $t_0$, $u_0$, $\sigma$, and $\tau$ such that
$P_0 \subseteq P$,
\(
t = t_0\tau \pfromB{}{\RR_1}
t_0\sigma \pfromB{\smash{P_0}}{\RR_1}
s \rto_{\RR_2}
u_0\sigma = u
\),
and $\sigma \pto{}_{\RR_1} \tau$. Thus $t_0 \fromto^*_\CC u_0$
follows.  Therefore,
$t = t_0\tau \fromto_\CC^* u_0\tau \pfromB{}{\RR_1} u_0\sigma = u$,
and hence \pref{lem:key:join} holds.
\end{itemize}
\item
If $P = \{\epsilon\}$, $Q \nsubseteq \{ \epsilon \}$,
and $\Gamma$ is not orthogonal then
the proof is analogous to the last case.
\item
If $P \nsubseteq \{\epsilon\}$ and $Q \nsubseteq \{ \epsilon \}$ then
$s$, $t$, and $u$ can be written as $f(s_1,\dots,s_n)$, $f(t_1,\dots,t_n)$,
and $f(u_1,\dots,u_n)$ respectively, and
$\Gamma_i\colon t_i \pfromB{}{\RR_1} s_i \pto{}_{\RR_2} u_i$
holds for all $1 \leqslant i \leqslant n$.  For every peak $\Gamma_i$ the
induction hypothesis yields (i) or (ii).
If \pref{lem:key:join} holds for all $\Gamma_i$ then 
\pref{lem:key:join} is concluded for $\Gamma$.
Otherwise, some $\Gamma_i$ satisfies \pref{lem:key:peak}. By taking
$t' = f(s_1,\ldots,t_i,\ldots,s_n)$ and
$u' = f(s_1,\ldots,u_i,\ldots,s_n)$ we have
$t \pfromB{}{\RR_1} t' \fromB{\PP} s \to_\PP u' \pto{}_{\PP} u$.
From $t_i \to_\RR^* \cdot \fromBT{\RR}{*} u_i$ we obtain
$t' \to_\RR^* \cdot \fromBT{\RR}{*} u'$.
Hence $\Gamma$ satisfies \pref{lem:key:peak}.
\qedhere
\end{enumerate}
\end{proof}

The next theorem is a compositional confluence criterion based on
parallel critical pair systems.

\begin{thm}
\label{thm:pcpsc}
Let $\RR$ be a left-linear TRS and $\CC$ a confluent TRS with 
$\CC \subseteq \RR$.  The TRS $\RR$ is confluent if
${\pcp{\RR}} \subseteq {\to_\RR^* \cdot \fromBT{\RR}{*}}$ and
$\PP/\RR$ is terminating, where $\PP = \PCPS(\RR, \CC)$.
\end{thm}
\begin{proof}
Let $\bot$ be a fresh symbol and let $I = \TT(\FF,\VV) \cup \{\bot\}$.  We
define the relation $>$ on $I$ as follows: $\alpha > \beta$ if
$\alpha \neq \bot = \beta$ or
$\alpha \to_{\PP/\RR}^+ \beta$.
Since $\PP/\RR$ is terminating, 
$>$ is a well-founded order.  Let $\AA = (\TT(\FF,\VV), \{ \pto{}_\alpha \}_{\alpha \in I})$ be
the ARS where $\pto{}_\alpha$ is defined as follows:
$s \pto{}_\alpha t$ if 
either $\alpha = \bot$ and $s \pto{}_\CC t$, or $\alpha \neq \bot$ and
$\alpha \to_\RR^* s \pto{}_{\RR \setminus \CC} t$.
Since the commutation of $\CC$ and $\CC$ follows from confluence of $\CC$,
\cref{lem:tait-martin-loef} yields
the commutation of $\to_\bot$ and $\to_\bot$. According to
\cref{lem:tait-martin-loef}
and \cref{thm:ddc}, it is
sufficient to show that every local peak 
\[
\smash{\Gamma\colon t \pfrom[\alpha]{} s \pto[\beta]{} u}
\]
with $(\alpha, \beta) \in I^2 \setminus \{ (\bot,\bot) \}$ is decreasing.
By the definition of $\AA$ we have $s \pto{}_{\RR_1} t$ and
$s \pto{}_{\RR_2} u$ for some TRSs
$\RR_1,\RR_2 \in \{\RR \setminus \CC, \CC\}$.
Using \cref{lem:key}, we distinguish two cases.
\begin{enumerate}
\item
Suppose that \cref{lem:key}\pref{lem:key:join} holds for $\Gamma$.
Then $t \pto{}_{\RR_2} t' \fromto^*_{\CC} u' \pfromB{}{\RR_1} u$
holds for some $t'$ and $u'$.
If $\RR_2 = \RR \setminus \CC$ then
$t \pto{}_\beta t'$ follows from
$\beta \to^*_\RR s \to^*_\RR t \pto{}_{\RR \setminus \CC} t'$.  Otherwise,
$\RR_2 = \CC$ yields $t \pto{}_\bot t'$.
In either case $t \pto{}_{\{\beta, \bot\}} t'$ is obtained. Similarly,
$u \pto{}_{\{\alpha,\bot\}} u'$ is obtained.
Moreover,
$t' \pfromto{}^*_{\bot} u'$ follows from $t' \fromto^*_\CC u'$.
Since $(\alpha,\beta) \neq (\bot,\bot)$
yields $\bot \in \curlyvee\alpha\beta$ and 
the reflexivity of $\pto{}_\bot$ yields
\(
{\pto{}_{\{\delta,\bot\}}} \subseteq
{\pto{}_\delta^= \cdot \pto{}_\bot}
\)
for any $\delta$, we obtain the desirable conversion
\(
t \pto[\beta]{=}
t' \pfromto[\curlyvee\alpha\beta]{*}
u' \pfrom[\alpha]{=}
u
\).
Hence, $\Gamma$ is decreasing.

\item
Suppose that \cref{lem:key}\pref{lem:key:peak} holds for $\Gamma$.  We have
$t \pfromB{}{\RR_1} t' \fromB{\PP} s \to_\PP u' \pto{}_{\RR_2} u$ and
$t' \to_\RR^* v \fromBT{\RR}{*} u'$ for some $t'$, $u'$, and $v$. As
$(\alpha,\beta) \neq (\bot,\bot)$, 
we have $\alpha \to^*_\RR s \to_\PP t'$ or
$\beta \to^*_\RR s \to_\PP t'$, from which $\alpha > t'$ or $\beta > t'$
follows.  Thus, $t' \in \curlyvee\alpha\beta$.
If $\RR_2 = \RR \setminus \CC$ then $t' \pto{}_{t'} t$.  Otherwise,
$\RR_2 = \CC$ yields $t' \pto{}_\bot t$.  So in either case 
$t' \pto{}_{\curlyvee\alpha\beta} t$ holds.
Next, we show $t' \pfromto{}_{\curlyvee \alpha\beta}^* v$.
Consider terms $w$ and $w'$ with
$t' \to_\RR^* w \to_\RR^{} w' \to_\RR^* v$.
We have $w \pto{}_{t'} w'$ or $w \pto{}_\bot w'$. So
$w \pto{}_{\curlyvee\alpha\beta} w'$ follows by
$\{t',\bot\} \subseteq \curlyvee\alpha\beta$.  
Summing up, we obtain
\(
t \pfrom{}_{\curlyvee\alpha\beta}
t' \pfromto{}_{\curlyvee\alpha\beta}^*
v
\).
In a similar way 
$u \pfrom{}_{\curlyvee\alpha\beta} u'
\pfromto{}_{\curlyvee\alpha\beta}^* v$ is
obtained.
Therefore
\(
\smash{
t \pfrom[\curlyvee\alpha\beta]{}
t' \pto[\curlyvee\alpha\beta]{*}
v \pfrom[\curlyvee\alpha\beta]{*}
u' \pto[\curlyvee\alpha\beta]{}
u
}
\),
and hence $\Gamma$ is decreasing.
\qedhere
\end{enumerate}
\end{proof}

We claim that \cref{thm:HM11} is subsumed by \cref{thm:pcpsc}.  Suppose
that $\CC$ is the empty TRS. Trivially $\CC$ is confluent.  Because
$\PCPS(\RR,\CC)$ is a subset of $\PCPS(\RR)$, termination of
$\PCPS(\RR,\CC)/\RR$ follows from that of $\PCPS(\RR)/\RR$, which
is equivalent to termination of $\CPS(\RR)/\RR$.  
Finally, ${\pcp{\RR}} \subseteq {\to_\RR^* \cdot \fromBT{\RR}{*}}$ 
is a necessary condition of confluence.
Thus, whenever \cref{thm:HM11} applies, 
\cref{thm:pcpsc} applies.

\cref{thm:pcpsc} also subsumes \cref{thm:SH22a}.
Suppose that $\CC$ is a confluent subsystem of $\RR$.  If
${\pcp{\RR}} \subseteq {\fromto^*_\CC}$ then
$\PCPS(\RR,\CC) = \varnothing$, which leads to termination of
$\PCPS(\RR,\CC)/\RR$.  Hence, \cref{thm:pcpsc} applies.  Note that 
if $\CC = \RR$ then $\PCPS(\RR,\CC) = \varnothing$.

\begin{exa}
\label{ex:cpsc}
Consider the left-linear TRS $\RR$:
\begin{alignat*}{6}
1\colon &\;& \m{s}(\m{p}(x)) & \to x \qquad
&
3\colon &\;& x + 0 & \to x 
&
5\colon &\;& x + \m{s}(y) & \to \m{s}(x + y) \qquad
\\
2\colon && \m{p}(\m{s}(x)) & \to x
&
4\colon && 0 + x & \to x + 0 \qquad
&
6\colon && x + \m{p}(y) & \to \m{p}(x + y)
\end{alignat*}
We show the confluence of $\RR$ by the combination of \cref{thm:pcpsc} and
orthogonality.
Let $\CC = \{ 3 \}$.
The TRS $\PCPS(\RR,\CC)$ consists of the eight rules:
\begin{align*}
\phantom{1\colon \;}
\m{0} + \m{s}(x) & \to \m{s}(\m{0} + x)
&
\phantom{1\colon \;}
x + \m{s}(\m{p}(y)) & \to \m{s}(x + \m{p}(y))
\\
\m{0} + \m{s}(x) & \to \m{s}(x) + \m{0}
&
x + \m{s}(\m{p}(y)) & \to x + y
\\
\m{0} + \m{p}(x) & \to \m{p}(\m{0} + x)
&
x + \m{p}(\m{s}(y)) & \to \m{p}(x + \m{s}(y))
\\
\m{0} + \m{p}(x) & \to \m{p}(x) + \m{0}
&
x + \m{p}(\m{s}(y)) & \to x + y
\end{align*}
The termination of $\PCPS(\RR,\CC)/\RR$ can be shown by, e.g.,
the termination tool \NaTT.
Since $\CC$ is orthogonal and
all parallel critical pairs of $\RR$ are joinable by $\RR$,
\cref{thm:pcpsc}
applies.
Note that the confluence of $\RR$ can neither be shown by
\cref{thm:ZFM15} nor \cref{thm:HM11}.  The former fails due to
the lack of suitable labeling functions for the following diagrams:
\begin{center}
\begin{tikzpicture}[D,baseline=(s)]
\node (s) at (0.0,2.0) {$x + \m{s}(\m{p}(y))$};
\node (t) at (0.0,0.0) {$x + y$};
\node (u) at (3.3,2.0) {$\m{s}(x + \m{p}(y))$};
\node (v) at (3.3,0.0) {$\m{s}(\m{p}(x + y))$};
\draw[->]
  (s) edge[-||->] node[left=1mm]  {$\{2\}$} node[right=1mm] {$1$} (t)
  (s) edge node[above] {$\epsilon$} node[below] {$5$} (u)
;
\draw[->,dashed]
  (u) edge node[right] {$6$} (v)
  (v) edge node[below] {$1$} (t)
;
\end{tikzpicture}
\hfil
\begin{tikzpicture}[D,baseline=(s)]
\node (s) at (0.0,2.0) {$x + \m{p}(\m{s}(y))$};
\node (t) at (0.0,0.0) {$x + y$};
\node (u) at (3.3,2.0) {$\m{p}(x + \m{s}(y))$};
\node (v) at (3.3,0.0) {$\m{p}(\m{s}(x + y))$};
\draw[->]
  (s) edge[-||->] node[left=1mm]  {$\{2\}$} node[right=1mm] {$2$} (t)
  (s) edge node[above] {$\epsilon$} node[below] {$6$} (u)
;
\draw[->,dashed]
  (u) edge node[right] {$5$} (v)
  (v) edge node[below] {$2$} (t)
;
\end{tikzpicture}
\end{center}
The latter fails due to the non-termination of $\CPS(\RR)/\RR$.
The culprit is the rule $\m{0} + \m{0} \to \m{0} + \m{0}$
in $\CPS(\RR)$, originating from the critical peak
$\m{0} \from \m{0} + \m{0} \to \m{0} + \m{0}$.
In contrast, the rule does not belong to $\PCPS(\RR,\CC)$
because the conversion $\m{0} \fromto^*_\CC \m{0} + \m{0}$ holds.
\end{exa}

Unlike the case of rule labeling, successive application of \cref{thm:pcpsc}
is more powerful than \cref{thm:HM11}.
\begin{exa}
\label{ex:CC}
By successive application of \cref{thm:pcpsc} we prove the confluence of the
left-linear TRS $\RR$:
\begin{alignat*}{4}
1\colon &\;&  \m{0} + x &\to x
\hspace{8em} &
3\colon && (x + y) + z &\to x + (y + z)
\\
2\colon && \m{s}(x) + y &\to \m{s}(x+y)
&
4\colon &\;&  x + (y + z) &\to (x + y) + z
\end{alignat*}
Let $\CC = \{1,2,3\}$.  Since the inclusion
${\pcp{\RR}} \subseteq {\to_\CC^* \cdot \fromBT{\CC}{*}}$ holds,
all parallel critical pairs of $\RR$ are joinable and
$\PCPS(\RR,\CC) = \varnothing$.  From the latter the termination of
$\PCPS(\RR,\CC)/\CC$ follows.  So it remains to show that $\CC$ is
confluent.
Since ${\pcp{\CC}} \subseteq {\to^*_\CC \cdot \fromBT{\CC}{*}}$ holds
and the termination of $\PCPS(\CC,\varnothing)/\CC$ follows from that of
$\CC$ (which is easily shown by the lexicographic path order~\cite{KL80}),
the confluence of $\CC$ follows from that of the empty TRS $\varnothing$.
Hence, $\RR$ is confluent.
Note that the confluence of $\RR$ cannot be shown by \cref{thm:HM11}
because $\CPS(\RR)/\RR$ is not terminating due to the rules 
of $\CPS(\RR)$:
\begin{align*}
x + ((y + z) + w) & \to (x + (y + z)) + w
&
(x + (y + z)) + w & \to x + ((y + z) + w) 
\end{align*}
\end{exa}

\section{Reduction Method}
\label{sec:reduction}

We present a \emph{reduction method} for confluence 
analysis.  The method shrinks a rewrite system $\RR$ to 
a subsystem $\CC$
such that $\RR$ is confluent iff $\CC$ is confluent.  Because
compositional confluence criteria address the `if' direction, the
question here is how to guarantee the reverse direction.  
In this section we develop a simple criterion, which exploits the fact
that confluence is preserved under signature extensions.  The resulting
reduction method can easily be automated by using SAT solvers.

We will show that if TRSs $\RR$ and $\CC$ satisfy
$\RR{\restriction}_\CC \subseteq {\to_\CC^*}$ then
confluence of $\RR$ implies confluence of $\CC$.  Here
$\RR{\restriction}_\CC$ stands for the following subsystem of $\RR$:
\[
\RR{\restriction}_\CC =
\{ \ell \to r \in \RR \mid \Fun(\ell) \subseteq \Fun(\CC) \}
\]
The following auxiliary lemma explains the role of the condition
$\RR{\restriction}_\CC \subseteq {\to_\CC^*}$.

\begin{lem} \label{lem:simulation}
Suppose $\RR{\restriction}_\CC \subseteq {\to_\CC^*}$.
\begin{enumerate}
\item
If $s \to_\RR t$ and 
$s \in \TT(\Fun(\CC), \VV)$
then $s \to_\CC^* t$ and 
$t \in \TT(\Fun(\CC), \VV)$
\item
If $s \to_\RR^* t$ and $s \in \TT(\Fun(\CC), \VV)$
then $s \to_\CC^* t$.
\end{enumerate}
\end{lem}
\begin{proof}
We only show the first claim, because then the second claim is shown by
straightforward induction. Suppose $s \in \TT(\Fun(\CC),\VV)$ and
$s \to_\RR t$.  There exist a rule $\ell \to r \in \RR$, a position
$p \in \Pos_\FF(s)$, and a substitution $\sigma$ such that $s|_p = \ell\sigma$
and $t = s[r\sigma]_p$.  As $s \in \TT(\Fun(\CC),\VV)$ implies
$\Fun(\ell) \subseteq \Fun(\CC)$, the rule $\ell \to r$ belongs to
$\RR{\restriction}_\CC$, which leads to $\ell \to_\CC^* r$ by assumption.
Since $\to_\CC^*$ is a rewrite relation, we obtain
$s = s[\ell\sigma]_p \to_\CC^* s[r\sigma]_p = t$.  The membership condition
$t \in \TT(\Fun(\CC),\VV)$ follows from $s \in \TT(\Fun(\CC),\VV)$ and
$s \to_\CC^* t$.
\qedhere
\end{proof}

As a consequence of \cref{lem:simulation}(2), confluence of $\RR$
carries over to confluence of $\CC$, when the inclusion
$\RR{\restriction}_\CC \subseteq {\to_\CC^*}$ holds and the signature of
$\CC$ is $\Fun(\CC)$.  The restriction against the signature of $\CC$ can
be lifted by the fact that confluence is preserved under signature
extensions:

\begin{prop}
\label{prop:signature-extension}
A TRS $\CC$ is confluent if and only if the
implication
\[
t \fromBT{\CC}{*} s \to_\CC^* u \implies 
t \to_\CC^* \cdot \fromBT{\CC}{*} u
\]
holds for all terms $s,t,u \in \TT(\Fun(\CC),\VV)$.
\end{prop}
\begin{proof}
Toyama~\cite{T87} showed that the confluence property is \emph{modular},
i.e., the union of two TRSs $\RR_1$ and $\RR_2$ over signatures $\FF_1$ and
$\FF_2$ with $\FF_1 \cap \FF_2 = \varnothing$ is confluent if and only if
both $\RR_1$ and $\RR_2$ are confluent.  Let $\CC$ be a TRS over a
signature $\FF$.  The claim follows by taking $\RR_1 = \CC$, 
$\RR_2 = \varnothing$, $\FF_1 = \Fun(\CC)$, and 
$\FF_2 = \FF \setminus \FF_1$.
\end{proof}

Now we are ready to show the main claim.

\begin{thm} \label{thm:converse}
Suppose $\RR{\restriction}_\CC \subseteq {\to_\CC^*}$.
If $\RR$ is confluent then $\CC$ is confluent.
\end{thm}
\begin{proof}
Suppose that $\RR$ is confluent. It is enough to show the implication in
\cref{prop:signature-extension} for all $s,t,u \in \TT(\Fun(\CC),\VV)$.
Suppose $t \fromBT{\CC}{*} s \to^*_\CC u$.  By confluence of $\RR$ we have
$t \to^*_\RR v \fromBT{\RR}{*} u$ for some $v$.
Since $\Fun(t)$ and $\Fun(u)$ are included in $\Fun(\CC)$,
\cref{lem:simulation} yields $t \to^*_\CC v \fromBT{\CC}{*} u$.
\end{proof}

A reduction method can be obtained by combining a compositional
confluence criterion with \cref{thm:converse}.   Here we present the
combination of \cref{thm:SH22a} with \cref{thm:converse} and its
automation technique.

\begin{cor} \label{cor:reduction}
Let $\CC$ be a subsystem of a left-linear TRS $\RR$ such that
${\pcp{\RR}} \subseteq {\fromto_\CC^*}$ and
$\RR{\restriction}_\CC \subseteq {\to_\CC^*}$.  The TRS $\RR$ is confluent
if and only if $\CC$ is confluent.
\end{cor}

The following example illustrates how \cref{cor:reduction} is used 
for automating confluence analysis.

\begin{exa}
\label{ex:reduction}
We show the confluence of the following left-linear TRS $\RR$:
\begin{align*}
1\colon \;
x + \m{0} & \to x
&
3\colon \;
\phantom{\m{s}(x) \times \m{0}}\makebox[0mm][r]{$\m{0} + y$} & \to y
&
5\colon \;
\m{s}(x) + y & \to \m{s}(x + y)
\\
2\colon \;
x \times \m{0} & \to \m{0}
&
4\colon \;
\m{s}(x) \times \m{0}  & \to \m{0}
&
6\colon \;
\m{s}(x) \times y  & \to (x \times y) + y
\end{align*}
Applying the reduction method of \cref{cor:reduction} repeatedly,
we remove rules unnecessary for confluence analysis.
\begin{enumerate}
\item
The TRS $\RR$ has four non-trivial parallel critical pairs and they admit
the following diagrams:
\begin{center}
\begin{tikzpicture}[D,baseline=(s)]
\node (s) at (1,1) {$\m{s}(x)+\m{0}$};
\node (t) at (0,0) {$\m{s}(x)$};
\node (u) at (2,0) {$\m{s}(x+\m{0})$};
\draw[->]
 (s) edge[-||->] (t)
 (s) edge node[above right] {$\epsilon$} (u)
 (u) edge[bend left] node[below] {$1$} (t)
;
\end{tikzpicture}
\quad
\begin{tikzpicture}[D,baseline=(s)]
\node (s) at (1,1) {$\m{s}(x)+\m{0}$};
\node (t) at (0,0) {$\m{s}(x+\m{0})$};
\node (u) at (2,0) {$\m{s}(x)$};
\draw[->]
 (s) edge[-||->] (t)
 (s) edge node[above right] {$\epsilon$} (u)
 (t) edge[bend right] node[below] {$1$} (u)
;
\end{tikzpicture}
\quad
\begin{tikzpicture}[D,baseline=(s)]
\node (s) at (1,1)  {$\m{s}(x) \times \m{0}$};
\node (t) at (0,0)  {$\m{0}$};
\node (u) at (2,0)  {$(x \times \m{0}) + \m{0}$};
\node (v) at (1,-1) {$x \times \m{0}$};
\draw[->]
 (s) edge[-||->] (t)
 (s) edge node[above right] {$\epsilon$} (u)
 (u) edge[bend left=10] node[right,yshift=-2mm] {$1$} (v)
 (v) edge[bend left=10] node[left,yshift=-2mm] {$2$} (t)
;
\end{tikzpicture}
\quad
\begin{tikzpicture}[D,baseline=(s)]
\node (s) at (1,1)  {$\m{s}(x) \times \m{0}$};
\node (t) at (0,0)  {$(x \times \m{0}) + \m{0}$};
\node (u) at (2,0)  {$\m{0}$};
\node (v) at (1,-1) {$x \times \m{0}$};
\draw[->]
 (s) edge[-||->] (t)
 (s) edge node[above right] {$\epsilon$} (u)
 (t) edge[bend right=10] node[left,xshift=-1mm] {$1$} (v)
 (v) edge[bend right=10] node[right,xshift=1mm] {$2$} (u)
;
\end{tikzpicture}
\end{center}
Therefore, ${\pcp{\RR}} \subseteq {\fromto_{\CC_0}^*}$ holds for 
$\CC_0 = \{1,2\}$.
As $\Fun(\CC_0) = \{\m{0}, {+}, {\times}\}$, we have 
$\RR{\restriction}_{\CC_0} = \{ 1,2,3 \}$.  However,
$\RR{\restriction}_{\CC_0} \subseteq {\to_{\CC_0}^*}$
does not hold due to $\m{0} + y \not\to_{\CC_0}^* y$.  So we 
extend $\CC_0$ to
$\CC = \CC_0 \cup \{ 3 \}$.  
Then
$\RR{\restriction}_\CC = \{1,2,3\} \subseteq {\to_\CC^*}$ holds.
Because $\CC$ is a superset of $\CC_0$, the inclusion 
${\pcp{\RR}} \subseteq {\fromto_\CC^*}$ 
holds too.
According to \cref{cor:reduction}, the confluence problem of $\RR$ is
reduced to that of $\CC$.  
\item
Since $\CC$ only admits a trivial parallel critical pair, it is closed by
the empty system $\varnothing$.
Moreover,
the inclusion
\(
\CC{\restriction}_\varnothing = \varnothing \subseteq {\to_\varnothing^*}
\)
holds. Hence, by \cref{cor:reduction} the confluence of $\CC$ is
reduced to the confluence of the empty system $\varnothing$.
\item
The confluence of the empty system $\varnothing$ is trivial.
\end{enumerate}
Hence we conclude that $\RR$ is confluent.
Note that in the first step all subsystems $\CC'$ including $\CC_0$ or
$\{1,4,6\}$ satisfy the inclusion
${\pcp{\RR}} \subseteq {\fromto_{\CC'}^*}$ but
some of them (e.g., $\{1,4,6\}$) are non-confluent.  The additional
requirement $\RR{\restriction}_{\CC'} \subseteq {\to_{\CC'}^*}$ excludes
such subsystems.

\end{exa}

\cref{cor:reduction} can be automated as follows.  Suppose that we have
found a subsystem $\CC_0$ of a given left-linear TRS $\RR$ such that
${\pcp{\RR}} \subseteq {\fromto_{\CC_0}^*}$.  We extend $\CC_0$ to $\CC$ so
that
(i)~$\CC_0 \subseteq \CC \subsetneq \RR$ and 
(ii)~$\RR{\restriction}_\CC \subseteq {\to_\CC^{\leqslant k}}$
for a designated number $k \in \NN$.  This search problem can be reduced to
a SAT problem. Let $\SSS_k(\ell \to r)$ be the following set of
subsystems:
\[
\SSS_k(\ell \to r) = \{ \{ \seq{\beta} \} \mid
\text{$\ell \to_{\beta_1} \cdots \to_{\beta_n} r$ and
$n \leqslant k$} \}
\]
In our SAT encoding we use two kinds of propositional variables: 
$x_{\ell \to r}$ and $y_f$.  The former represents $\ell \to r \in \CC$,
and the latter represents $f \in \Fun(\CC)$.  With these variables
the search problem for $\CC$ is encoded as follows:
\[
\bigwedge_{\alpha \in \CC_0} x_\alpha
\;\land\;
\bigvee_{\alpha \in \RR} 
\lnot x_\alpha
\;\land\;
\bigwedge_{\alpha \in \RR} 
\biggl(\lnot x_\alpha \lor
\bigwedge_{f \in \Fun(\alpha)} y_f
\biggr)
\;\land\;
\bigwedge_{\alpha \in \RR \setminus \CC_0}
\biggl(
\bigl(
\bigvee_{\SS \in \SSS_k(\alpha)}
x_\SS
\bigr)
\;\lor\;
\bigl(
\lnot \bigwedge_{f \in \Fun(\ell)} y_f
\bigr)
\biggr)
\]
Here $x_\SS = x_{\beta_1} \land \cdots \land x_{\beta_n}$ for
$\SS = \{\seq{\beta}\}$.  It is easy to see that the first two clauses
encode condition (i) and the third clause characterizes $\Fun(\CC)$. The
last clause encodes condition (ii). 

\begin{exa}[Continued from~\cref{ex:reduction}]
Recall that ${\pcp{\RR}} \subseteq {\fromto_{\CC_0}^*}$ holds for
$\CC_0 = \{1,2\}$. Setting $k = 5$, we compute $\SSS_k(\alpha)$ for each
rule $\alpha \in \RR \setminus \CC_0 = \{3,4,5,6\}$:
\begin{align*}
\SSS_k(3) & = \{\{3\}\}
&
\SSS_k(4) & = \{\{2\},\{1,2,6\},\{2,3,6\}\}
&
\SSS_k(5) & = \{\{5\}\}
&
\SSS_k(6) & = \{\{6\}\}
\end{align*}
The SAT encoding explained above results in the following formula
\[
\begin{array}{@{}l@{~}c@{~}l@{~}c@{~}l@{~}c@{~}l@{}}
(x_1 \,\land\, x_2)  
& \land & (\lnot x_1 \,\lor \cdots \lor \lnot x_{6})
& \land & (\lnot x_1 \,\lor\, (y_{\m{0}} \,\land\, y_+)) 
\\[0.2em]
&&
& \land & (\lnot x_2 \,\lor\, (y_{\m{0}} \,\land\, y_\times))
\\[0.2em]
&&
& \land & (\lnot x_3 \,\lor\, (y_{\m{0}} \,\land\, y_+)) 
& \land & (x_3 \,\lor\, \lnot (y_{\m{0}} \,\land\, y_+))
\\[0.2em]
&&
& \land & (\lnot x_4 \,\lor\, (y_{\m{0}} \,\land\, y_{\m{s}} \,\land\, y_\times))
& \land & 
(X \;\lor\, \lnot (y_{\m{s}} \,\land\, y_{\m{0}} \,\land\, y_\times))
\\[0.2em]
&&
& \land & (\lnot x_5 \,\lor\, (y_{\m{s}} \,\land\, y_{+}))
& \land & (x_5 \,\lor\, \lnot (y_{\m{s}} \,\land\, y_+))
\\[0.2em]
&&
& \land & (\lnot x_6 \,\lor\, (y_{\m{s}} \,\land\, y_+ \,\land\, y_\times))
& \land & (x_6 \,\lor\, \lnot (y_{\m{s}} \,\land\, y_\times))
\end{array}
\]
with
\(
X = x_2 \lor (x_1 \land x_2 \land x_6) \lor (x_2 \land x_3 \land x_6) 
\).
The formula is satisfied if we assign true to $x_1$,
$x_2$, $x_3$, $y_{\m{0}}$, $y_+$, and $y_\times$, and false to the other
variables.  This assignment corresponds to $\CC = \{1,2,3 \}$.
Note that for this formula there is no other solution.
\end{exa}

\section{Experiments}
\label{sec:experiments}

In order to evaluate the presented approach we implemented 
a prototype confluence tool \OURTOOL which supports
the main three compositional confluence criteria
(Theorems~\ref{thm:SH22a}, \ref{thm:prlc}, and \ref{thm:pcpsc}) and their
original versions (Theorems~\ref{thm:R73}, \ref{thm:ZFM15}, and \ref{thm:HM11})
as well as the reduction method (\cref{cor:reduction}).\footnote{%
The tool and the experimental data are available at
\url{https://www.jaist.ac.jp/project/saigawa/}. These are also available
at \cite{SH23}.%
}
The problem set used in experiments consists of
$462$ left-linear TRSs taken from the confluence
problems database COPS~\cite{HNM18}.  Out of the $462$ TRSs, at least
$190$ are known to be non-confluent.
The tests were run on a PC with Intel Core i7-1065G7 CPU (1.30 GHz) and 16
GB memory of RAM using timeouts of $120$ seconds.  
\cref{tbl:experiments} summarizes the results.
The columns in the table stand for the following confluence criteria:
\begin{itemize}
\item
\criterion{O}:
Orthogonality (\cref{thm:R73}).
\item
\criterion{R}:
Rule labeling (\cref{thm:ZFM15}).
\item
\criterion{C}:
The criterion by critical pair systems (\cref{thm:HM11}).
\item
\criterion{OO}:
Successive application of \cref{thm:SH22a},
as illustrated in \cref{ex:SH22a}.
\item
\criterion{CC}:
Successive application of \cref{thm:pcpsc},
as illustrated in \cref{ex:CC}.
\item
\criterion{RC}:
\cref{thm:prlc}, where confluence of a subsystem $\CC$
is shown by \cref{thm:pcpsc} with the empty subsystem.
\item
\criterion{CR}:
\cref{thm:pcpsc}, where confluence of a subsystem $\CC$ is shown by
\cref{thm:prlc} with the empty subsystem.
\item
\criterion{rOO},
\criterion{rRC}, and
\criterion{rCR}:
The combination of the reduction method (\cref{cor:reduction}) with
\criterion{OO}, \criterion{RC}, and \criterion{CR}, respectively.
\item
\OURTOOL:
The combination of the reduction method with
\criterion{RC} and \criterion{CR}.
\end{itemize}
\begin{table}[t]
\centering
\caption{Experimental results on $462$ left-linear TRSs.}
\setlength{\tabcolsep}{0.3em}
\begin{tabular}{@{}lrrrrrrrrrrrrrrr@{}}
\toprule
& \criterion{O}
& \criterion{R}
& \criterion{C}
& \criterion{OO}
& \criterion{CC}
& \criterion{RC}
& \criterion{CR}
& \criterion{rOO}
& \criterion{rCC}
& \criterion{rRC}
& \criterion{rCR}
& \OURTOOL
& \ACP
& \CoLL
& \CSI
\\[0.25em]
proved
& $20$
& $135$
& $59$
& $88$
& $111$
& $152$
& $143$
& $91$
& $114$
& $153$
& $146$
& $154$
& $197$
& $194$
& $216$
\\[0.25em]
timeouts
& $0$
& $20$
& $10$
& $13$
& $68$
& $88$
& $42$
& $10$
& $59$
& $81$
& $49$
& $79$
& $51$
& $156$
& $4$
\\
\bottomrule
\end{tabular}
\label{tbl:experiments}
\end{table}
Note that in any combination the reduction method is successively
applied, as in \cref{ex:reduction}.
For the sake of comparison the results of the confluence tools \ACP
version 0.72~\cite{AYT09}, 
\CoLLSaigawa version 1.7~\cite{SH15},
and \CSI version 1.2.7~\cite{ZFM11a}
are also included in the table,
where \CoLLSaigawa is abbreviated to \CoLL.

We briefly explain how these criteria are automated in our tool.  Suitable
subsystems for the compositional criteria are searched by enumeration.
Relative termination, required by Theorems~\ref{thm:HM11} and
\ref{thm:pcpsc}, is checked by employing the termination tool \NaTT version
2.3~\cite{YKS14}.
Joinability of each (parallel) critical pair $(t, u)$ is tested by the relation:
\[
t
\xrightarrow{\leqslant 5}
\cdot \xleftarrow{\leqslant 5}
u
\]
For rule labeling, the decreasingness
of each parallel critical peak $t \pfromB{P}{{\phi,k}} s \rto_{\psi,m} u$
is checked by existence of a conversion of the form
\[
t
\xrightarrow[\curlyvee k]{}^{i_1}
\cdot \pto[\psi,m]{}^{i_2}
\cdot \xrightarrow[\curlyvee km]{}^{i_3} 
\cdot \mathrel{^{j_3}{\xleftarrow[\curlyvee km]{}}}
v \mathrel{^{j_2}{\pfrom[\phi,k]{P'}}}
\cdot \mathrel{^{j_1}{\xleftarrow[\curlyvee m]{}}}
u
\]
such that $i_1,i_3,j_1,j_3 \in \NN$,
$i_2, j_2 \in \{0,1\}$, $i_1+i_2+i_3 \leqslant 5$,
$j_1+j_2+j_3 \leqslant 5$, and the inclusion
$\Var(v,P') \subseteq \Var(s,P)$ holds.
This is encoded into linear arithmetic constraints~\cite{HM11},
and they are solved by the SMT solver \textsf{Z3} version~4.8.11~\cite{dMB08}.
Finally, automation of the reduction method~(\cref{cor:reduction}) is
done by SAT solving as presented in \secref{sec:reduction}.
To organize it as a lightweight method, we test only
one combination of join sequences.
The SMT solver \textsf{Z3} is used for solving SAT problems for the method.

As theoretically expected, in the experiments \criterion{O} is subsumed by
both \criterion{R} and \criterion{C}.
The results of \criterion{OO} and \criterion{CC} clearly show
effectiveness of successive application,\footnote{Successive application
of rule labeling is same as \criterion{R}, see~\secref{sec:prlc}.}
while \criterion{OO} is subsumed by \criterion{R} and \criterion{CC}.
Concerning the combinations of \criterion{R} and \criterion{C},
the union of \criterion{R} and
\criterion{C} amounts to $145$, and the union of \criterion{RC} and
\criterion{CR} amounts to $153$.  
Due to timeouts, \criterion{CR} misses three systems of which \criterion{R}
can prove confluence.
Differences between \criterion{RC} and \criterion{CR} are summarized as
follows:
\begin{itemize}
\item
Three systems are proved by \criterion{RC} but not by 
\criterion{CR}.\footnote{%
The three systems are COPS numbers \texttt{994}, \texttt{1001}, and
\texttt{1029}.}
One of them is the next TRS (COPS number \texttt{994}).  \criterion{RC}
uses the subsystem $\{2,4,6\}$ whose confluence is shown by
$\criterion{C}$.
\begin{align*}
1\colon\;\m{a}(\m{b}(x)) & \to \m{a}(\m{c}(x))
&
3\colon\;\m{c}(\m{b}(x)) & \to \m{a}(\m{b}(x))
&
5\colon\;\m{c}(\m{c}(x)) & \to \m{c}(\m{c}(x))
\\
2\colon\;\m{a}(\m{c}(x)) & \to \m{c}(\m{b}(x))
&
4\colon\;\m{b}(\m{c}(x)) & \to \m{a}(\m{c}(x))
&
6\colon\;\m{c}(\m{c}(x)) & \to \m{c}(\m{b}(x))
\\
&&&&
7\colon\; \m{c}(\m{b}(x)) & \to \m{a}(\m{b}(x))
\end{align*}

\item
The only TRS where \criterion{CR} is advantageous to \criterion{RC} is
COPS number \texttt{132}:
\begin{align*}
1\colon\;& -(x + y) \to (-x) + (-y)
&
3\colon\;& -(-x) \to x
\\
2\colon\;& (x + y) + z \to x + (y + z)
&
4\colon\;& x + y \to y + x
\end{align*}
Its confluence is shown by the composition of \cref{thm:pcpsc} and
\cref{thm:ZFM15}, the latter of which proves the subsystem $\{1,2,4\}$
confluent.
\end{itemize}

The columns \criterion{rOO}, \criterion{rRC}, and \criterion{rCR} in
\cref{tbl:experiments} show that the use of the reduction method
(\cref{cor:reduction}) basically improves the power and efficiency of the
underlying compositional confluence criteria. 
Our observations on the results are as follows:

\begin{itemize}
\item
For $106$ systems the reduction method removed at least one rule.  Out of
these $106$ systems, $55$ were reduced to the empty system.
While the use of the reduction method as a preprocessor improves the
efficiency in most of cases, there are a few exceptions (e.g., COPS
number~\texttt{689}).  The bottleneck is the reachability test by
${\to_{\CC^{\leqslant k}}}$.

\item
The confluence proving powers of \criterion{rOO} and \criterion{OO} are
theoretically equivalent, because the reduction method as a compositional
confluence criterion is an instance of \criterion{OO}.
In the experiments
\criterion{rOO} handled three more systems.  This is due to the
improvement of efficiency.  
The same argument holds for the relation between \criterion{rRC} and
\criterion{RC}.

\item
The reduction method and \criterion{C} are incomparable with each
other.  Hence \criterion{rCR} is more powerful than \criterion{CR}.
In the experiments, \criterion{rCR} subsumes
\criterion{CR} and it includes three more systems.  As a drawback,
\criterion{rCR} has seven more timeouts.

\item
Among \criterion{rOO}, \criterion{rRC}, and \criterion{rCR}, the second
criterion is the most powerful.  
As in the cases of their underlying criteria, 
the results of \criterion{rOO} are 
subsumed by both \criterion{rRC} and \criterion{rCR}, and
COPS number~\texttt{132} is the only problem where \criterion{rCR}
outperforms \criterion{rRC}.
\end{itemize}
\OURTOOL is the union of \criterion{rRC} and \criterion{rCR}.  Although
the number is behind those of the state-of-art tools, 
the number contains a system (\texttt{COPS number 1001}) that is
handled only by \OURTOOL (due to \criterion{RC}).

Finally, we discuss how the results of the other confluence tools
change if the reduction method is used as their preprocessor:
\begin{itemize}
\item
\ACP gains three proofs but also misses three proofs based on
\emph{reduction-preserving completion}~\cite[Definition 4.7]{AT12}.
While this technique uses a subsystem $\PP$ with
${\to_{\PP}} \subseteq {\fromBT{\PP}{*}}$, in the three proofs
the reduction method virtually shrinks $\PP$ to $\varnothing$.
Although \ACP does not use reduction-preserving completion with 
$\PP = \varnothing$, if \ACP does, the proofs are recovered.
\item
\CoLLSaigawa increases the number to $201$, gaining 7 proofs.
\item
\CSI
gains no proofs. Since the tool supports
rule labeling
(\criterion{R}), it can partly cover the class of problems that the
reduction method is effective.  Moreover, the tool employs
\emph{redundant rule elimination}~\cite{NFM15,SH15}, which plays a
similar role to the reduction method. In the next section we will
discuss this elimination method as related work.
\end{itemize}

\section{Conclusion}
\label{sec:related-work}

We studied how compositional confluence criteria can be derived from
confluence criteria based on the decreasing diagrams technique, and
showed that Toyama's almost parallel closedness theorem is subsumed by his
earlier theorem based on parallel critical pairs.  We conclude
the paper by mentioning related work and future work.

\paragraph*{Simultaneous critical pairs.}
van~Oostrom~\cite{vO97} showed the almost development closedness theorem:
A left-linear TRS is confluent if the inclusions
\begin{align*}
{\xleftarrow{\epsilon}}{\rtimes}{\xrightarrow{\epsilon}}
& \subseteq {\xrightarrow{*} \cdot \mfrom{}}
&
{\xleftarrow{>\epsilon}}{\rtimes}{\xrightarrow{\epsilon}}
& \subseteq {\mto{}}
\end{align*}
hold, where $\mto{}$ stands for the
multi-step~\cite[Section~4.7.2]{TeReSe}.  Okui~\cite{O98} showed the
simultaneous closedness theorem: A left-linear TRS is confluent if the
inclusion
\[
{{\mfrom{}}{\rtimes}{\xrightarrow{}}}
\subseteq
{\xrightarrow{*} \cdot \mfrom{}}
\]
holds, where ${\mfrom{}}{\rtimes}{\xrightarrow{}}$ stands for the set of
simultaneous critical pairs~\cite{O98}.  As this inclusion characterizes
the inclusion
\(
{\mfrom{} \cdot \to} \subseteq {\to^* \cdot \mfrom{}}
\),
simultaneous closedness subsumes almost development closedness.
The main result in \secref{sec:pc} is considered as a
counterpart of this relationship in the setting of parallel critical pairs.

\paragraph*{Critical-pair-closing systems.}
A TRS $\CC$ is called \emph{critical-pair-closing} for a TRS $\RR$ if
\[
{\cp{\RR}} \subseteq {\fromto_\CC^*}
\]
holds.  It is known that a left-linear TRS $\RR$ is confluent if $\CCd/\RR$
is terminating for some confluent critical-pair-closing TRS $\CC$ with
$\CC \subseteq \RR$, see~\cite{HNvOO19}.  Here $\CCd$ denotes the set of
all duplicating rules in $\CC$.
\cref{thm:SH22a} imposes closedness by
$\CC$ on all \emph{parallel} critical pairs in return to removal of the
relative termination condition.
Investigating whether the latter subsumes the former is our future work.

\paragraph*{Rule labeling.}
Dowek et al.~\cite[Theorem~38]{DFJL22} extended rule labeling based on parallel
critical pairs~\cite{ZFM15} to take higher-order rewrite systems.  If we
restrict their method to a first-order setting, it corresponds to 
the case that a complete TRS is employed for $\CC$ in \cref{thm:prlc},
and thus, it can be seen as a generalization of 
\cref{cor:T17} by Toyama~\cite{T17}.

\paragraph*{Critical pair systems.}
The second author and Middeldorp~\cite{HM13} generalized
\cref{thm:HM11} by replacing $\CPS(\RR)$ by the following subset:
\[
\CPS'(\RR) = \{ s \to t, s \to u \mid
\text{$t \fromB{\RR} s \rto_\RR u$ is a 
critical peak but
not $t \mto{}_\RR u$} \}
\]
This variant subsumes van Oostrom's development closedness
theorem~\cite{vO97}.  We anticipate that in a similar way our compositional
variant (\cref{thm:pcpsc}) is extended to subsume the parallel closedness
theorem based on parallel critical pairs (\cref{thm:T81}).

\paragraph*{Redundant rules.}
\emph{Redundant rule elimination} by Nagele 
et al.~\cite[Corollary 9]{NFM15} can be regarded as a compositional
confluence criterion.  It states that a TRS $\RR$ is confluent if there
exists a confluent subsystem $\CC$ such that 
$\RR \setminus \CC \subseteq {\fromto_\CC^*}$ holds.
When $\RR$ is left-linear, the criterion is subsumed by \cref{thm:SH22a}.
This is verified by the following trivial fact:
\begin{fact}
Let $\CC$ be a subsystem of a TRS $\RR$.
If $\RR \setminus \CC \subseteq {\fromto_\CC^*}$ then
${\pcp{\RR}} \subseteq {\fromto_\CC^*}$.
\end{fact}
The converse does not hold in general.  To see it, consider the one-rule
TRS $\RR$ consisting of $\m{a} \to \m{b}$. The empty TRS $\CC = \varnothing$
satisfies ${\pcp{\RR}} \subseteq {\fromto_\CC^*}$ but 
$\RR \setminus \CC \subseteq {\fromto_\CC^*}$ does not hold
as $\m{a} \mathrel{\centernot{\fromto}_\CC^*} \m{b}$.
There is another form of redundant rule elimination
(\cite[Corollary 6]{NFM15} and \cite{SH15}). It states that a TRS $\RR$ is
confluent if and only if $\RR \subseteq {\to_\CC^*}$ for some confluent
$\CC \subseteq \RR$.
This criterion is regarded as a reduction method for confluence
analysis. In fact, it is an instance of \cref{cor:reduction} for
left-linear TRSs, since $\RR{\restriction}_\CC \subseteq \RR$ and 
${\pcp{\RR}} \subseteq {\fromto_\CC^*}$ hold.  We want to stress that a
reduction method is obtained by any combination of a compositional
confluence criterion with \cref{thm:converse}.

\paragraph*{Modularity and automation.}
Last but not least, we discuss 
relations between modularity and
reduction methods.  Organizing compositional criteria as a reduction method
is a key for effective automation.  Therefore, developing a generalization
of \cref{thm:converse} is our primary future work.
Ohlebusch~\cite{O02} showed that if the union of \emph{composable} TRSs
$\RR$ and $\CC$ is confluent then both $\RR$ and $\CC$ are confluent.  When
$\CC$ is a subsystem of $\RR$, this result is rephrased as follows: If
$\DD_{\RR \setminus \CC} \cap \Fun(\CC) = \varnothing$ then confluence of
$\RR$ implies that of $\CC$. Therefore, this can be used as an alternative of
\cref{thm:converse}.  
Unfortunately, $\RR{\restriction}_\CC \subseteq \CC$ follows
from $\DD_{\RR \setminus \CC} \cap \Fun(\CC) = \varnothing$. So
composability as a reduction method is still in the realm of our
criterion (\cref{thm:converse}).  Similarly, we can argue that
the theorem also subsumes the persistency result~\cite{AT97} as a
base criterion for reduction methods.  Yet, we anticipate that this
work benefits from studies of more advanced modularity results such
as \emph{layer systems}~\cite{FMZvO15}.
Another future work is to develop an effective confluence analysis
based on compositional confluence criteria and reduction methods.  The use
of the \emph{confluence framework}~\cite{GVL22} which exploits modularity
results would be worth investigating.

\section*{Acknowledgment}
We are grateful to Jean-Pierre Jouannaud, Vincent van Oostrom,
and Yoshihito Toyama for their valuable comments on
preliminary results of this work.  
We are also grateful to Ren\'e Thiemann for spotting and correcting a
mistake in the proof of \cref{thm:prlc} in the preliminary version of this
paper~\cite{SH22}.
Last but not least, we thank the reviewers of 
this article and its preliminary version~\cite{SH22}
for their thorough reading and suggestions, which greatly
helped to improve the presentation.

\bibliographystyle{alphaurl}
\bibliography{references}
\end{document}